\newif\ifTwoColumn
\newif\ifTechReport

\TechReporttrue    

\ifTwoColumn
\documentclass[journal]{IEEEtran}
\else
\ifTechReport
\documentclass[journal,12pt,onecolumn,draftcls]{IEEEtran}
\else
\documentclass[journal,11pt,onecolumn,draftclsnofoot]{IEEEtran}
\fi
\fi

\usepackage[utf8]{inputenc} 
\usepackage[T1]{fontenc}    
\usepackage{url}            
\usepackage{booktabs}       
\usepackage{amsfonts,amsmath,amsthm,amssymb,mathrsfs}       
\usepackage{bbm}
\usepackage{cite}
\usepackage{mathtools,xparse}
\usepackage[makeroom]{cancel}
\usepackage[acronym]{glossaries}
\usepackage{color}
\usepackage{balance}
\usepackage{algpseudocode}
\usepackage{subcaption}
\usepackage{nicefrac}       
\usepackage{microtype}      
\usepackage{xcolor}         
\usepackage{graphicx}
\usepackage[prependcaption,colorinlistoftodos]{todonotes}
\usepackage{multirow}
\usepackage{booktabs,tabularx}
\usepackage{cases}
\usepackage[ruled,vlined]{algorithm2e}
\usepackage{varwidth}
\usepackage{tikz}
	\usetikzlibrary{fit,arrows,calc,positioning,shapes.geometric}

\graphicspath{{figures/}}

\renewcommand{\qedsymbol}{$\blacksquare$}

\DeclareMathSizes{10}{10}{6}{4}

\newcommand{\R}{\mathbb{R}}

\newcommand{\N}{\mathbb{N}}

\newcommand{\bbS}{\mathbb{S}}
\newcommand{\mc}[1]{\mathcal{#1}}

\newcommand{\eqdef}{\coloneqq}
\newcommand{\reqdef}{\eqqcolon}

\newcommand{\rank}{\mathrm{rank}}
\newcommand{\vvec}{\mathrm{vec}}

\newcommand{\col}{\mathrm{col}}

\newcommand{\bs}[1]{\boldsymbol{#1}}
\newcommand{\bsone}{\boldsymbol{1}}

\newtheorem{theorem}{Theorem}[section]
\newtheorem{definition}[theorem]{Definition}
\newtheorem{proposition}[theorem]{Proposition}
\newtheorem{lemma}[theorem]{Lemma}

\newtheorem{remark}[theorem]{Remark}
\newtheorem{standing}[theorem]{Standing Assumption}

\newcommand{\normaltext}[1]{\textnormal{#1}}

\makeglossaries
\newacronym{LP}{LP}{linear program}
\newacronym{LTI}{LTI}{linear time-invariant}
\newacronym{CLF}{CLF}{control Lyapunov function}
\newacronym{MPC}{MPC}{model predictive control}
\newacronym{eMPC}{eMPC}{explicit model predictive control}
\newacronym{LQR}{LQR}{linear quadratic regulator}
\newacronym{mp-QP}{mp-QP}{multi-parametric quadratic program}
\newacronym{QP}{QP}{quadratic program}
\newacronym{PWA}{PWA}{piecewise-affine}
\newacronym{MIP}{MIP}{mixed-integer program}
\newacronym{ReLU}{ReLU}{rectified linear unit}
\newacronym{PWA-NN}{PWA-NN}{piecewise-affine neural network}
\newacronym{NN}{NN}{neural network}
\newacronym{MI}{MI}{mixed-integer}
\newacronym{MILP}{MILP}{mixed-integer linear program}
\newacronym{LICQ}{LICQ}{linear independence constraint qualification}
\newacronym{KKT}{KKT}{Karush-Kuhn-Tucker}
\newacronym{ISS}{ISS}{input-to-state stable}

\tikzstyle{startstop} = [rectangle, rounded corners, minimum width=2.5cm, text width=.8\columnwidth,  minimum height=1cm, text centered, draw=black, fill=red!30]
\tikzstyle{io} = [trapezium, trapezium left angle=70, trapezium right angle=110, minimum width=1.5cm, minimum height=.75cm, text centered, draw=black, fill=blue!30]
\tikzstyle{process} = [rectangle, minimum width=1cm, minimum height=1cm, text width=.8\columnwidth, text centered, draw=black, fill=orange!30]
\tikzstyle{decision} = [diamond, minimum width=3cm, minimum height=3cm, text centered, draw=black, fill=green!30, ]
\tikzstyle{arrow} = [thick,->,>=stealth]
\tikzstyle{c} = [rectangle, draw, inner sep=0.5cm, dashed]
\tikzstyle{line} = [draw, -latex']

\begin{document}

\title{Reliably-stabilizing piecewise-affine\\ neural network controllers}
\author{Filippo Fabiani and Paul J. Goulart
	\thanks{The authors are with the Department of Engineering Science, University of Oxford, OX1 3PJ, United Kingdom ({\tt \{filippo.fabiani, paul.goulart\}@eng.ox.ac.uk}).
		This work was partially supported through the Government’s modern industrial strategy by Innovate UK, part of UK Research and Innovation, under Project LEO (Ref. 104781).}}


\maketitle

\begin{abstract}
	A common problem affecting \gls{NN} approximations of \gls{MPC} policies is the lack of analytical tools to assess the stability of the closed-loop system under the action of the \gls{NN}-based controller. We present a general procedure to quantify the performance of such a controller, or to design minimum complexity \glspl{NN} with \glspl{ReLU} that preserve the desirable properties of a given \gls{MPC} scheme. By quantifying the approximation error between \gls{NN}-based and \gls{MPC}-based state-to-input mappings, we first establish suitable conditions involving two key quantities, the worst-case error and the Lipschitz constant, guaranteeing the stability of the closed-loop system. We then develop an offline, mixed-integer optimization-based method to compute those quantities exactly. Together these techniques provide conditions sufficient to certify the stability and performance of a \gls{ReLU}-based approximation of an \gls{MPC} control law.
\end{abstract}


\IEEEpeerreviewmaketitle

\glsresetall

\section{Introduction}\label{sec:intro}
\IEEEPARstart{M}{\lowercase{odel predictive control}} (MPC) is one of the most popular control strategies for linear systems with operational and physical constraints \cite{rawlings2017model,borrelli2017predictive} and is based on the repeated solution of constrained optimal control problems. In its \emph{implicit} version, the optimal sequence of control inputs is computed by solving an optimization problem \emph{online} that minimizes some objective function subject to constraints, taking the current state as an initial condition. For certain classes of problems, e.g., when the system dynamics and constraints are affine and the cost function is quadratic (resulting in a \gls{QP}), solutions to these optimization problems can also be pre-computed \emph{offline} using multi-parametric programming \cite{johansen2000explicit,bemporad2002explicit}, with the initial state condition serving as the parameter. This allows one to instead implement an \emph{explicit} version of an MPC policy, i.e., \gls{eMPC}, which amounts to implementing a \gls{PWA} control law. In this latter case, the online computational requirements reduce to identifying over some polyhedral partition of the state-space the region in which the system state resides, e.g., via a binary search tree, and implementing the associated control action.

Although the theory underlying MPC is quite mature and provides practical stability, safety and performance guarantees, implementations of both implicit and explicit MPC suffer from intrinsic practical difficulties. In the case of implicit MPC, the computational effort of solving optimization problems in real-time complicates its application in systems characterized by very high sampling rates \cite{qin2003survey}, which feature prominently in emerging applications in robotics \cite{erez2013integrated,nubert2020safe}, aerial and autonomous vehicles \cite{zhang2016learning,varshney2019deepcontrol,zhu2020safe,richter2018bayesian}, and power electronics \cite{maddalena2021embedded}.

Conversely, \gls{eMPC} requires far less real-time computation, but complications arise when applied to systems with even moderate state dimensions or on embedded systems with modest computational and memory resources \cite{johansen2014toward}. This is because the complexity of the associated \gls{PWA} controller, measured by the number of affine pieces and regions, is known to grow exponentially with the state dimension and number of constraints \cite{alessio2009survey}, making it intractable for large systems. Moreover, generating an associated search tree to determine which region contains the current state may fail or lead to gigantic lookup tables, thus requiring too much processing power or memory storage for online evaluation \cite{kvasnica2011clipping}.

As a result, the idea of \emph{approximating} an MPC policy, using various  techniques, traces back more than 20 years ago \cite{parisini1995receding,bemporad2003suboptimal,jones2009approximate}. The use of (deep) \glspl{NN} \cite{hagan1997neural,Goodfellow-et-al-2016} is particularly attractive in view of their universal approximation capabilities, typically requiring a relatively small number of parameters \cite{hornik1989multilayer}. Despite their computationally demanding offline training requirements, the online evaluation of \gls{NN}-based approximations to MPC laws is computationally very inexpensive, since it only requires the evaluation of an input-output mapping \cite{duarte2018fast,zhang2019real,schindler2020real}.
However, unless one assumes a certain structure, deep \glspl{NN} are generally hard to analyze due to their nonlinear and large-scale structure \cite{hagan1997neural,Goodfellow-et-al-2016}.

\subsection{Related works}
For the aforementioned reasons, along with the growing interest in data-driven control techniques generally, interest in \gls{NN}-based approximations of MPC laws has increased rapidly in recent years \cite{chen2018approximating,hertneck2018learning,karg2020efficient,karg2020stability,zhang2020near,maddalena2020neural,paulson2020approximate}. Starting from the pioneering work in \cite{parisini1995receding}, where a \gls{NN} with one hidden layer was adopted to learn a constraint-free nonlinear MPC policy through fully supervised learning, \cite{chen2018approximating} proposed to train a \gls{NN} with \glspl{ReLU} using a reinforcement learning method, resulting in an efficient and computationally tractable training phase. In \cite{karg2020efficient} it was shown that \gls{ReLU} networks can encode exactly the \gls{PWA} mapping resulting from the formulation of an MPC policy for \gls{LTI} systems, with theoretical bounds on the number of hidden layers and neurons required for such an exact representation.  However, these approaches do not come equipped with certificates of closed-loop stability.

Conversely, the approach in \cite{karg2020stability} performed the output reachability analysis of a \gls{ReLU}-based controller via a collection of \gls{MILP} formulations (one for each control and state constraint) to establish closed-loop asymptotic stability requirements involving the underlying \gls{NN}.
In \cite{hertneck2018learning} a robust MPC for a deterministic, nonlinear system was first designed to tolerate inaccurate input approximations up to a certain tolerance, and then a \gls{NN} was trained to mimic such a robust scheme. By leveraging statistical methods, probabilistic stability and constraint satisfaction guarantees for the closed-loop system were then shown to be possible. Along the same lines, \cite{zhang2020near} focused on linear parameter-varying systems for which a \gls{NN}-based approximation of an MPC policy enjoyed probabilistic guarantees of feasibility and near-optimality. By considering \gls{LTI} systems with an additive source of uncertainty, \cite{paulson2020approximate} proposed to approximate a robust MPC scheme with a \gls{NN} and then project the output of such an approximation into a suitably chosen set, in order to guarantee robust constraint satisfaction and stability of the closed-loop system. Finally, \cite{maddalena2020neural} developed a method to fit samples from an MPC law by means of a tailored \gls{NN} built as a composition of linear maps and optimization problems. While the stability analysis of the \gls{LTI} system with such a peculiar \gls{NN} approximation can be carried out via sum-of-squares programming, the controller deployment requires either one forward pass in the \gls{NN} or the offline computation of two additional \gls{PWA} mappings.

\subsection{Summary of contributions}
In contrast to the existing literature, we provide a means to assess the training quality of a \gls{ReLU} network in replicating the action of an MPC policy. Specifically, we establish a systematic, \gls{MI} optimization-based procedure that allows us to certify the reliability, in terms of stability and performance of the closed-loop system, of a \gls{ReLU}-based approximation of an MPC law. In summary, we make the following contributions:

\begin{itemize}
	\item By considering the approximation error between a \gls{ReLU} network and an MPC law, we give sufficient conditions involving the maximal approximation error and the associated Lipschitz constant that guarantee the closed-loop stability of a discrete-time \gls{LTI} system when the \gls{ReLU} approximation replaces the MPC law;
	\item We formulate an \gls{MILP} to compute the Lipschitz constant of an MPC policy exactly. This result is of standalone interest, as well as being instrumental for the main result;
	\item We develop an optimization-based technique to allow the exact computation of the worst-case approximation error and the Lipschitz constant characterizing the approximation error.  The outcome is a set of conditions involving the optimal value of two \glspl{MILP} that are sufficient to allow us to certify the reliability of the \gls{ReLU}-based controller;
	\item We suggest several ways to employ our results in practice.
\end{itemize}

This work represents a first step towards a unifying theoretical framework for analyzing NN-based approximations of MPC policies, since we are able to assess the closed-loop stability of an \gls{LTI} system with a \gls{PWA-NN} controller that results from the training of the network to mimic a given MPC law. This approach was not proposed in any of the relevant work on \gls{NN}-based control  \cite{parisini1995receding,chen2018approximating,hertneck2018learning,karg2020efficient,karg2020stability,zhang2020near,maddalena2020neural,paulson2020approximate}.
Moreover, compared to more traditional suboptimal MPC schemes \cite{bemporad2003suboptimal,zeilinger2011real,necoara2013linear,giselsson2013feasibility,rubagotti2014stabilizing}, the design of minimum complexity \gls{ReLU}-based approximations  moves the computational requirements completely offline, i.e., the training phase, strategies for input constraints satisfaction and certificates verification, and does not need any artificial problem modifications that may degrade control performance, e.g., constraints tightening. For these reasons, \gls{PWA-NN} controllers are suitable candidates to maintain the optimality features of an MPC scheme with inexpensive online evaluation \cite{duarte2018fast,zhang2019real,schindler2020real}.

\subsection{Paper organization}
In \S \ref{sec:problem_description} we present the approximation problem addressed in the paper, whereas in \S \ref{sec:stability} we establish closed-loop stability criteria to motivate the interest in some quantities characterizing the approximation error. Successively, \S \ref{sec:background} introduces some mathematical ingredients needed for our treatment, while \S \ref{sec:maximal_gain} is devoted to establish a preliminary result involving the Lipschitz constant computation of an MPC policy. In \S \ref{sec:training_quality}, we give the main result characterizing the exact computation of the key quantities discussed in \S \ref{sec:stability}, while \S \ref{sec:discussion} reports a general procedure suggesting how to use of our results. We accompany this section with a discussion of practical aspects, including accommodation of input constraints, and we finally verify our theoretical results via a numerical example involving the stabilization of a system of coupled oscillators in \S \ref{sec:simulations}.

\subsection*{Notation}
$\N$, $\R$ and $\R_{\geq 0}$ denote the set of natural, real and nonnegative real numbers, respectively. $\N_0 \eqdef \N \cup \{0\}$, $\N_\infty \eqdef \N \cup \{+\infty\}$, while $\mathbb{B} \eqdef \{0,1\}$. $\bbS^{n}$ is the space of $n \times n$ symmetric matrices and $\bbS_{\succcurlyeq 0}^{n}$ is the cone of positive semi-definite matrices. A vector with all elements equal to $1$ ($0$) is denoted by a bold $\bsone$ ($\bs{0}$). Given a matrix $A \in \R^{m \times n}$, $A^\top$ denotes its transpose, $a_{i,j}$ its $(i,j)$ entry, $A_{:,j}$ (resp., $A_{i,:}$) its $j$-th column ($i$-th row). Given a vector $v \in \R^m$, for any set of indices $\mc{I} \subseteq \{1, \ldots, m\}$, $A_{\mc{I}}$ (resp., $v_{\mc{I}}$) denotes the submatrix (subvector) obtained by selecting the rows (elements) indicated in $\mc{I}$. $A \otimes B$ represents the Kronecker product between matrices $A$ and $B$. For $A \in \bbS_{\succcurlyeq 0}^{n}$, $\| v \|_A \eqdef \sqrt{ v^\top A v }$. Given vectors $v$, $u \in \R^m$, $v \perp u$ imposes a complementarity condition between them, i.e., $v^\top u = 0$. With $\|A\|_{\alpha} \eqdef \textrm{sup}_{\|x\|_\alpha \leq 1} \, \|A x\|_\alpha = \textrm{sup}_{x \neq 0} \, \|A x\|_\alpha / \|x\|_\alpha$ we denote the norm over matrices in $\R^{m \times n}$ induced by an arbitrary norm $\|\cdot\|_\alpha$ over both $\R^n$ and $\R^m$. Given a function $V : \R^n \to \R_{\geq 0}$, $\Omega_a \eqdef \{x \in \R^n \mid V(x) \leq a\}$ denotes the generic $a$-sublevel set of $V(\cdot)$. $\mathrm{exp} : \R \to \R_{> 0}$ represents the natural exponential function. For a given set $\mc{S} \subseteq \R^n$, $|\mc{S}|$ represents its cardinality, while $\textrm{int}(\mc{S})$ its topological interior. Given a mapping $F : \R^n \to \R^m$, the local $\alpha$-Lipschitz constant over some set $\mc{S} \subseteq \R^n$ is denoted as $\mc{L}_{\alpha}(F,\mc{S})$.
With a slight abuse of notation, we indicate with $\partial F(\mc{S})$ the generalized Jacobian of $F(\cdot)$ over the whole set $\mc{S}$. The operator $\col(\cdot)$ stacks its arguments in column vectors or matrices of compatible dimensions, $\textrm{avg}(\cdot)$ is the average operator of its arguments, $\vvec(\cdot)$ maps a matrix to a vector that stacks its columns, and $\textrm{proj}_{\mc S} : \R^m \to \mc{S}$ denotes the standard point-to-set projection mapping \cite[\S 8.1]{boyd2004convex}. To indicate the state evolution of discrete-time \gls{LTI} systems, we sometimes use $x(k+1)$, $k \in \N_0$, as opposed to $x^+$, making the time dependence explicit whenever necessary.

\glsresetall

\section{An approximation problem}\label{sec:problem_description}
We will consider the problem of stabilizing the constrained, discrete-time, \gls{LTI} system
\begin{equation}\label{eq:LTI_sys}
	x^+ = Ax + B u,
\end{equation}
with state variable $x \in \mc{X}$, control input $u \in \mc{U}$ and state-space matrices $A \in \R^{n \times n}$ and $B \in \R^{n \times m}$.  We will assume that the constraint sets $\mc{X} \subseteq \R^n$ and $\mc{U} \subseteq \R^m$ are bounded polyhedral.
A popular control choice for constrained systems is \gls{MPC}, an optimization-based control method implemented in receding horizon. Specifically, an \emph{implicit} \gls{MPC} policy requires one to solve, at every iteration, the following \gls{mp-QP} over a finite time horizon of length $T \geq 1$, $\mc{T} \eqdef \{0, \ldots, T-1\}$,
\begin{equation}\label{eq:QP_init}
	  V_T(x) =  \left\{
	\begin{aligned}
		&\underset{({v_i})_{i \in \mc{T}}}{\textrm{min}} &&  \tfrac{1}{2} \|x_T\|^2_P + \sum_{i \in \mc{T}} \tfrac{1}{2} (\|x_i\|^2_Q + \|v_i\|^2_R)\\
		&\hspace{.15cm}\textrm{ s.t. } && x_{i+1} = A x_i + B v_i, \;  i \in \mc{T},\\
		&&& x_i \in \mc{X}, \, v_i \in \mc{U}, \; i \in \mc{T}, \\
		&&&x_0 = x.
	\end{aligned}
	\right.
\end{equation}
Starting from some $x(0) \in \mc{X}$, the receding horizon implementation of an \gls{MPC} law computes an optimal solution $(v^\star_i)_{i \in \mc{T}}$, and then applies the control input $u(0) = v^\star_0$ taken from the first part of the optimal sequence. This process is then repeated at every time $k$ with initial condition $x_0 = x(k)$, so that the procedure amounts to the implicit computation of a fixed mapping $x \mapsto v^\star_0(x)$. We define this control law explicitly as
\[
u_\textrm{MPC}(x) \eqdef v_0^\star(x),
\]
to emphasize this dependence.
Under standard assumptions on the data characterizing the optimization problem in \eqref{eq:QP_init} (i.e., weight matrices and constraints), it is well known that the associated \gls{MPC} control law $u(k) = u_\textrm{MPC}(x(k))$ stabilizes the constrained \gls{LTI} system \eqref{eq:LTI_sys} about the origin \cite{borrelli2017predictive,rawlings2017model} while, at the same time, respecting state and input constraints.

In some applications the dynamics of the underlying system may be too fast relative to the  time required to compute the solution to the \gls{mp-QP} in \eqref{eq:QP_init}. One may then rely on the \emph{explicit} version of the \gls{MPC} law in \eqref{eq:QP_init}, i.e., \gls{eMPC}  \cite{bemporad2002explicit}, whose closed form expression can be computed offline.
The optimal solution mapping $u_\textrm{MPC}(\cdot)$ enjoys a \gls{PWA} structure that maps any $x \in \mc{X}$ into an affine control action according to some polyhedral partition of $\mc{X}$. The partition and associated affine functions for $u_\textrm{MPC}(\cdot)$ can be computed offline, e.g., using \gls{MPC} Toolbox \cite{bemporad2021model}.   However, the computational effort required for this offline computation may itself be too demanding, since the number of regions in the optimal partition can grow exponentially with the number of states and constraints in  \eqref{eq:QP_init} \cite{alessio2009survey}.  In addition, even if computable offline, the online implementation of the explicit solution may have excessive storage requirements.

The well-known limitations of \gls{MPC} and \gls{eMPC} motivate the design of an approximation for $u_\textrm{MPC}(\cdot)$ that can be implemented with minimal computation and storage requirements while still maintaining stability and good performance of the closed-loop system. We focus on controllers implemented using \gls{ReLU} neural networks \cite{Goodfellow-et-al-2016}, which provide a natural means for approximating $u_\textrm{MPC}(\cdot)$ since the output mapping of such a network has \gls{PWA} structure \cite{montufar2014number,siahkamari2020piecewise}.

Specifically, after training a \gls{ReLU} network to produce a mapping $u_\textrm{NN} : \mc{X} \to \R^m$ to approximate $u_\text{MPC}(\cdot)$, we ask whether the training was sufficient to ensure stability of the closed-loop system in \eqref{eq:LTI_sys} with \gls{PWA-NN} controller $u_\textrm{NN}(\cdot)$ in place of $u_\textrm{MPC}(\cdot)$ (see Fig.~\ref{fig:NN_feedback}). In the next section we describe the features of the approximation error function $e(x) \eqdef u_{\textrm{NN}}(x) - u_{\textrm{MPC}}(x)$ that are crucial to guarantee the stability of the closed-loop system in \eqref{eq:LTI_sys} with the \gls{PWA-NN} controller $u_\textrm{NN}(\cdot)$.  We subsequently provide an \gls{MI} optimization-based method to exactly compute those quantities offline. The result will be a set of conditions on the optimal value of two \glspl{MILP} sufficient to certify the stability of the closed-loop system \eqref{eq:LTI_sys} under the action of approximated \gls{MPC} law $u_\textrm{NN}(\cdot)$.

\begin{figure}[t!]
	\centering
	\ifTwoColumn
\else
		\includegraphics[width=.5\columnwidth]{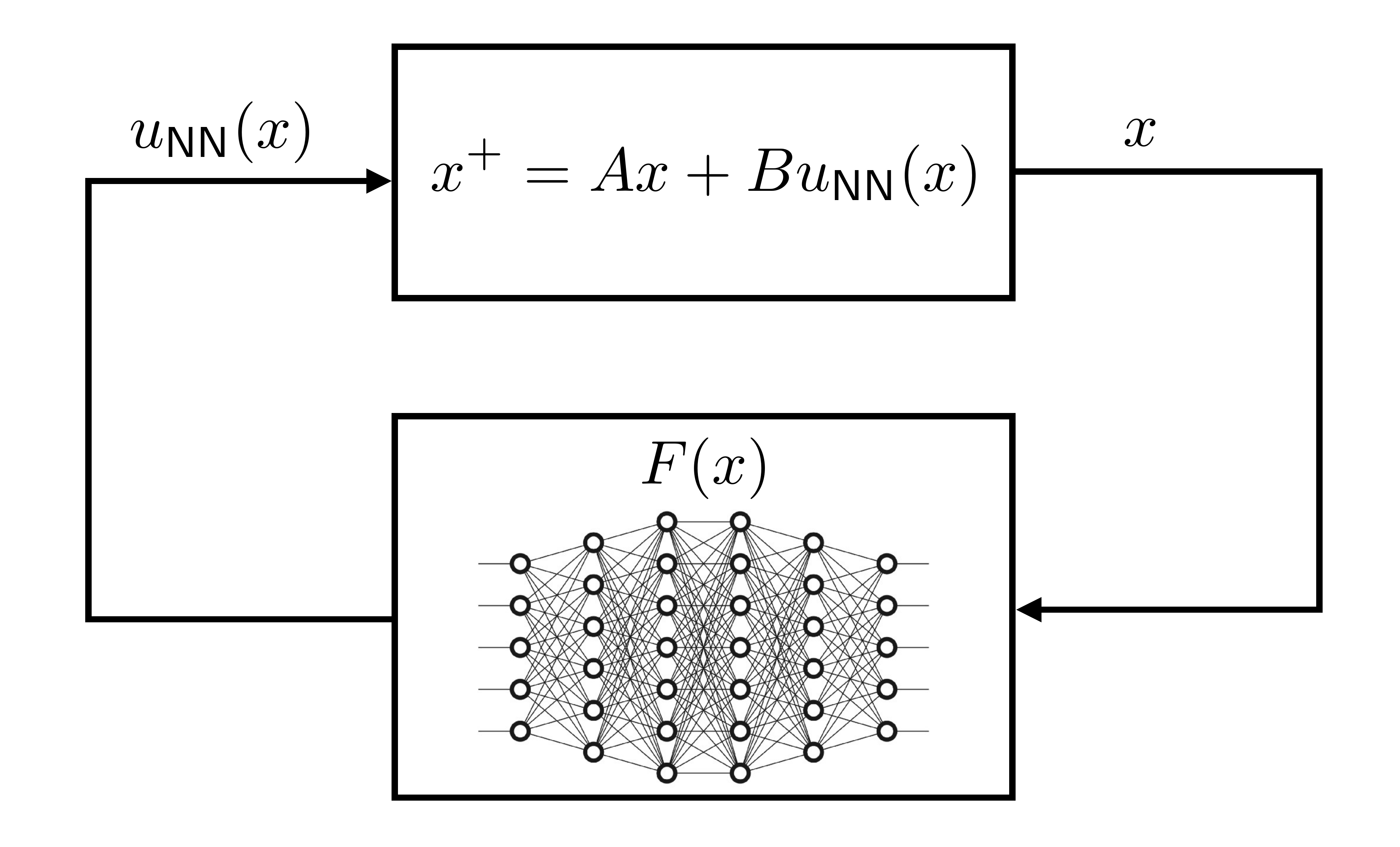}
	\fi
	\caption{Feedback loop with \gls{PWA-NN} controller.}
	\label{fig:NN_feedback}
\end{figure}

\section{Stability analysis of piecewise-affine\\ neural network controllers}\label{sec:stability}
We first address the stability of the \gls{LTI} system in \eqref{eq:LTI_sys} with approximately optimal controller $u = u_{\textrm{NN}}(x)$, by considering the robust stability of the underlying system with \gls{MPC} policy, $u_{\textrm{MPC}}(\cdot)$, subject to an additive disturbance:
\begin{equation}\label{eq:perturbed_dyn}
		x^+ = Ax + Bu_{\textrm{NN}}(x) = Ax + Bu_{\textrm{MPC}}(x) + Be(x).
\end{equation}

We assume that the approximation error $e : \mc{X} \to \R^m$ is bounded on $\mc{X}$ and Lipschitz continuous in some set $\mc{X}_\infty \subseteq \mc{X}$ to be defined later. For some $\alpha \in \N_\infty$, we hence assume that there exist constants $\bar{e}_\alpha$, $\mc{L}_\alpha(e, \mc{X}_\infty) \geq 0$ such that $\|e(x)\|_\alpha \leq \bar{e}_\alpha$ for all $x \in \mc{X}$, and $\|e(x) - e(y)\|_\alpha \leq \mc{L}_\alpha(e, \mc{X}_\infty) \|x - y\|_\alpha$ for all $x$, $y \in \mc{X}_\infty$. In \S \ref{sec:training_quality}, we will show how these conditions can be made to hold, providing an \gls{MI} optimization-based method to compute $\bar{e}_\alpha$ and $\mc{L}_\alpha(e, \mc{X}_\infty)$ exactly.
Throughout the paper, we make the following mild assumption:

\begin{standing}\label{standing:mp-QP}
	For the \normaltext{\gls{LTI}} system \eqref{eq:LTI_sys} under the action of the \normaltext{\gls{MPC}} policy $u_{\normaltext{\textrm{MPC}}}(\cdot)$:
	\begin{itemize}
		\item[(i)] the origin is exponentially stable;
		\item[(ii)] the \normaltext{\gls{mp-QP}} in \eqref{eq:QP_init} is recursively feasible starting from any $x \in \mc{X}$. \hfill$\square$
		\end{itemize}
\end{standing}
%

Since the system in \eqref{eq:LTI_sys} is \gls{LTI}, Standing Assumption~\ref{standing:mp-QP}.(i) is satisfied under standard design choices (see, e.g., \cite[\S 2.5--2.6]{rawlings2017model}), whereas condition (ii) is assumed without restrictions, as our results also apply if one considers a subset of $\mc{X}$ for which the problem in \eqref{eq:QP_init} is feasible.
By exploiting the optimal cost $V_T(\cdot)$ of the \gls{mp-QP} in \eqref{eq:QP_init}, with $\Omega_a$ denoting the associated $a$-sublevel set, we first establish that the closed-loop system in \eqref{eq:LTI_sys} with an approximated \gls{MPC} law is \gls{ISS} \cite{jiang2001input} when the maximal approximation error $\bar{e}_\alpha$ is sufficiently small. The proof of this result, along with the others in this section, is deferred to Appendix~\ref{app:proofs}.

\begin{lemma}\label{lemma:NN_stability}
	There exists $\zeta > 0$ such that, if $\bar{e}_{\alpha} < \zeta$, the \normaltext{\gls{LTI}} system in \eqref{eq:LTI_sys} with \normaltext{\gls{PWA-NN}} controller $u = u_{\normaltext{\textrm{NN}}}(x)$ converges exponentially to some neighbourhood of the origin $\Omega_b$, for all $x(0) \in \Omega_c \supset \Omega_b$, with $c \eqdef \textnormal{\textrm{max}} \, \{a \geq 0 \mid \Omega_a \subseteq \mc{X}\}$.
	\hfill$\square$
\end{lemma}

Lemma~\ref{lemma:NN_stability} says that if the worst-case approximation error over $\mc{X}$ is strictly smaller that $\zeta$, then the closed-loop in \eqref{eq:LTI_sys} with \gls{PWA-NN} controller $u_{\textrm{NN}}(\cdot)$ is \gls{ISS} and its state trajectories satisfy the constraints, since $\Omega_c$ is robust positively invariant.
From the related proof, it turns out that $\zeta$ is tunable through a nonnegative parameter $\rho$, which strikes a balance between the robustness of the closed-loop system and the performance of the approximated controller (see Appendix~\ref{app:proofs}, specifically \eqref{eq:error_bound}). In fact, the larger the $\rho$, the larger the approximation error that can be tolerated while guaranteeing \gls{ISS}. On the other hand, this reduces the guaranteed rate of convergence to $\Omega_b$. 

Define $\mc{X}_\infty$ as the set of states for which the stabilizing unconstrained linear gain (typically the \gls{LQR}), $\bar{K}_\textrm{MPC} \in \R^{m \times n}$, satisfies both state and control constraints, i.e., $\mc{X}_\infty \eqdef \{x \in \mc{X} \mid x(0) = x \in \mc{X}, A x(k) + B \bar{K}_\textrm{MPC}x(k) \in \mc{X}, \, \bar{K}_\textrm{MPC}x(k) \in \mc{U}, \, k \in \mc{T}, \, x(T) \in \mc{X}\}$. Within $\mc{X}_\infty$, which is the maximal output admissible set as described in \cite{gilbert1991linear}, the system \eqref{eq:perturbed_dyn} still enjoys exponential convergence if the local Lipschitz constant of $e(\cdot)$ meets a certain condition:

\begin{lemma}\label{lemma:NN_stability_local}
	There exists $\vartheta > 0$ such that, if $\mc{L}_\alpha(e,\mc{X}_\infty) < \vartheta$, the \normaltext{\gls{LTI}} system in \eqref{eq:LTI_sys} with \normaltext{\gls{PWA-NN}} controller $u = u_{\normaltext{\textrm{NN}}}(x)$ converges exponentially to the origin for all $x(0) \in \mc{X}_\infty$.
	\hfill$\square$
\end{lemma}

Putting the previous results together gives us our main stability result, upon which subsequent requirements on the fidelity of our \gls{ReLU}-based controllers will be based:

\begin{theorem}\label{th:exp_conv}
	If $\bar{e}_\alpha$ and $\mc{L}_\alpha(e, \mc{X}_\infty)$ satisfy \eqref{eq:error_bound} and \eqref{eq:upper_bound_lipschitz}, respectively, and $b \geq 0$ can be chosen so that $\Omega_b \subseteq \mc{X}_\infty$, then the \normaltext{\gls{LTI}} system in \eqref{eq:LTI_sys} with \normaltext{\gls{PWA-NN}} controller $u = u_{\normaltext{\textrm{NN}}}(x)$ converges exponentially  to the origin, for all $x(0) \in \Omega_c$.
	\hfill$\square$
\end{theorem}

\begin{remark}\label{remark:omega_b}
	A less conservative condition still guaranteeing exponential stability is possible by replacing  $\mc{L}_\alpha(e, \mc{X}_\infty)$ with $\mc{L}_\alpha(e, \Omega_b)$ in \eqref{eq:upper_bound_lipschitz},  but requires the availability of $\bar{e}_\alpha < \zeta$ to tune $b \ge 0$ properly. Computing $\mc{L}_\alpha(e, \mc{X}_\infty)$ can instead be done independently of the values of $b$ and $\bar{e}_\alpha$. However, we observe in \normaltext{\S \ref{sec:simulations}} that since $\bar{e}_\alpha$ can be made very small in practice,  $\Omega_b$ likewise reduces to a very small neighbourhood of the origin so that $\Omega_b \subset \mc{X}_\infty$. Verifying \eqref{eq:upper_bound_lipschitz} with $\mc{L}_\alpha(e, \Omega_b)$ is thus practically meaningful as it allows us to i) recover exponential stability and ii) reduce the computational time.
	\hfill$\square$
\end{remark}
For our application the certificates established in Lemma~\ref{lemma:NN_stability} and \ref{lemma:NN_stability_local} are sufficient only, and hence conservative: a trained \gls{ReLU} network that does not meet those certificates could indeed behave well in practice (see, e.g., Table~\ref{tab:masses_results} in \S \ref{sec:simulations} for those cases in  which $\mc{L}_\infty(e, \Omega_b) < \vartheta$ is not met).
Nevertheless, Lemma~\ref{lemma:NN_stability} is not conservative in the sense that some error $e(\cdot)$ whose norm exceeds $\zeta$ by any amount could conceivably force non-convergence if selected in an adversarial way (a similar statement applies to Lemma~\ref{lemma:NN_stability_local}).
Specifically, since the worst-case error $\bar e_\alpha$ is attained at a particular state in $\mc{X}$, to force non-convergence the worst-case error would not only need to exceed $\zeta$ in norm but also to be realized at the most disadvantageous location in the domain of the controller. While this is unlikely to happen in practice, we can not preclude the possibility that the error mapping $e(\cdot)$ will be somehow realized in a particularly disfavourable way.   If it were to be, then the result would still hold but would be 
non-conservative.

In the rest of the paper, we provide an \gls{MI} optimization-based method to compute the worst-case approximation error and the (local) Lipschitz constant of $e(\cdot)$ exactly, thus providing conditions sufficient to certify the stability and performance of a \gls{ReLU}-based approximation of an \gls{MPC} control law.


\section{Mathematical background}\label{sec:background}
We next consider some properties of both \glspl{PWA-NN} based on \gls{ReLU} networks and of \glspl{mp-QP} typically originating \gls{MPC} policies.  We start with the definition of a \gls{PWA} mapping:

\begin{definition}\textup{(Piecewise-affine mapping \cite[Def.~2.47]{rockafellar2009variational})}\label{def:pwa}
	A continuous mapping $F:\mc{F} \to \R^m$ is \emph{piecewise-affine} on the closed domain $\mc{F} \subseteq \R^n$ if
	\begin{itemize}
		\item[(i)] $\mc{F}$ can be partitioned on a finite union of $N$ disjoint polyhedral sets, i.e. $\mc{F} \eqdef \cup_{i =1}^N \mc{F}_i$ with $\mc{F}_i \cap \mc{F}_j = \emptyset$~$\forall i \neq j$;
		\item[(ii)] $F(\cdot)$ is affine on each of the sets $\mc{F}_i$, i.e.\
		$F(x) = F_i(x) \eqdef G_i x + g_i$, $G_i \in \R^{m \times n}$, $g_i \in \R^m$, $\forall x \in \mc{F}_i$. \hfill$\square$
	\end{itemize}

\end{definition}

From Definition~\ref{def:pwa}, any continuous \gls{PWA} mapping on $\mc{F}$ is also Lipschitz continuous according to the following definition:

\begin{definition}\textup{(Lipschitz constant)}
	The local $\alpha$\emph{-Lipschitz constant} of a mapping $F : \R^n \to \R^m$ over the set $\mc{F} \subseteq \R^n$ is
	\begin{equation}\label{eq:lip_gen}
		\mc{L}_{\alpha}(F,\mc{F}) \eqdef \underset{x \neq y \in \mc{F}}{\normaltext{\textrm{sup}}} \; \frac{\|F(x) - F(y)\|_{\alpha}}{\|x - y\|_{\alpha}}.
	\end{equation}
	If $\mc{L}_{\alpha}(F,\mc{F})$ exists and is finite, then we say that $F(\cdot)$ is $\alpha$-\emph{Lipschitz continuous} over the set $\mc{F}$.
	\hfill$\square$
\end{definition}

\begin{lemma}\textup{(\hspace{-.1mm}\cite[Prop.~3.4]{gorokhovik1994pwa})}\label{lemma:pwa_lip}
	Given an arbitrary norm $\|\cdot\|_\alpha$ over both $\R^n$ and $\R^m$, any \normaltext{\gls{PWA}} mapping $F:\mc{F} \to \R^m$ has $\alpha$-Lipschitz constant of $\normaltext{\textrm{max}}_{i = 1,\ldots,N} \, \|G_i^\top\|_{\alpha}$.
	\hfill$\square$
\end{lemma}

Lemma~\ref{lemma:pwa_lip} says that it is possible to compute exactly the Lipschitz constant of a \gls{PWA} mapping if one knows the linear term $G_i$ of every component of $F(\cdot)$.  Specifically, $\mc{L}_{\alpha}(F,\mc{F})$ coincides with the maximum gain over the partition of $\mc F$.


\subsection{A family of \gls{PWA} neural networks}\label{subsec:PLNN}
An $L$-layered, feedforward, fully-connected \gls{NN} that defines a mapping $F: \R^n \to \R^{m}$ can be described by the following recursive equations across layers \cite{hagan1997neural}:
\begin{equation}\label{eq:RELU_NN}
	\left\{
	\begin{aligned}
		& x^{0} = x,\\
		& x^{j +1} = \phi(W^j x^j + b^j), \; j \in \{0, \ldots, L-1\},\\
		& F(x) = W^L x^L + b^L,
	\end{aligned}
	\right.
\end{equation}
where $x^{0} = x \in \R^{n_0}$, $n_0 = n$, is the input to the network, $W^{j} \in \R^{n_{j + 1} \times n_j}$ and $b^j \in \R^{n_{j+1}}$ are the weight matrix and bias vector of the $(j+1)$-th layer, respectively (defined during some offline training phase).  The total number of neurons is thus $N \eqdef \sum_{j  = 1}^{L} n_{j} + m$, since $n_{L+1} = m$.
The \emph{activation function} $\phi:\R^{n_j} \to\R^{n_{j}}$ applies component-wise to the pre-activation vector $W^j x^j + b^j$, assumed identical for each layer.

%

Since we focus on \gls{ReLU} networks \cite{Goodfellow-et-al-2016}, we take the activation function to be $\phi(\cdot) = \textrm{max}(\cdot, 0)$. In this case, although $F(\cdot)$ is known to be a \gls{PWA} mapping \cite{montufar2014number,siahkamari2020piecewise}, an explicit description as in Definition~\ref{def:pwa} is generally difficult to compute, thereby complicating determination of its Lipschitz constant.

We next make an assumption that allows us to compute $\mc{L}_{\alpha}(F,\mc{X})$ exactly in case $\|\cdot\|_{\alpha}$ is a linear norm.
For a given $x \in \mc{X}$, let $f_i : \R^n \to \R$, $i \in \{1, \ldots, N\}$, be the input to the $i$-th neuron, i.e., \gls{ReLU}, of the network, and let $\mc{K}_i \eqdef \{x \in \mc{X} \mid f_i(x) = 0\}$ be the associated $i$-th \gls{ReLU} kernel of $F(\cdot)$. 

%

\begin{standing}\label{standing:NN_general_position}
	The \normaltext{\gls{ReLU}} network $F : \R^n \to \R^m$ in \eqref{eq:RELU_NN} is in \emph{general position} in the sense of \normaltext{\cite[Def.~4]{jordan2020exactly}}, i.e.\ for every subset of neurons $\mc{Q} \subseteq \{1, \ldots, N\}$, $\cap_{i \in \mc{Q}} \, \mc{K}_i$ is a finite union of $(n - |\mc{Q}|)$-dimensional polytopes.
	\hfill$\square$
\end{standing}

Given this assumption, the local Lipschitz constant $\mc{L}_{\alpha}(F,\mc{X})$ of a trained \gls{ReLU} network can be computed through an \gls{MILP} \cite[Th.~5]{jordan2020exactly}.  The assumption is widely used in the machine learning literature \cite{hanin2019deep,jordan2020exactly}, and it has been proven that almost every \gls{ReLU} network is in general position \cite[Th.~3]{jordan2020exactly}.

\subsection{\gls{MPC} and \gls{mp-QP} optimization}\label{subsec:MPC}
Since $\mc{X}$ and $\mc{U}$ in \eqref{eq:QP_init} are polytopic sets, the optimization problem in \eqref{eq:QP_init} can be rewritten as an equivalent \gls{mp-QP} with inequality constraints only. Specifically, it amounts to
\begin{equation}\label{eq:QP_final}
	\left\{
	\begin{aligned}
		&\underset{\bs{z}}{\textrm{min}} &&  \tfrac{1}{2} \bs{z}^\top H \bs{z}\\
		&\hspace{0.0cm}\textrm{ s.t. } && N \bs{z} \leq d + S x,
	\end{aligned}
	\right.
\end{equation}
where $\bs{z} \eqdef \bs{v} + H^{-1} D^\top x$, $\bs{v} \eqdef \col((v_i)_{i \in \mc{T}}) \in \R^{mT}$, $H \in \mathbb{S}^{mT}_{\succcurlyeq 0}$ and vector/matrices $H, D, N, d, S$ of appropriate dimensions are obtained from $Q$, $R$, $P$, $A$, $B$, $T$, and the data defining $\mc{X}$ and $\mc{U}$.
In particular, the state propagation constraint in \eqref{eq:QP_init}, $x_{i+1} = A x_i + B v_i$ with $x_0 = x$, allows us to write
 	$$
 		\begin{bmatrix}
 			x_0\\
 			x_1\\
 			x_2\\
 			\vdots\\
 			x_{T}
 		\end{bmatrix} = \underbrace{\begin{bmatrix}
 			0 & 0 & \cdots & 0\\
 			B & 0 & \cdots & 0\\
 			AB & 0 & \cdots & 0\\
 			\vdots &  \vdots & \ddots & \vdots\\
 			A^{T-1}B & A^{T-2} B & \cdots & B
	 	\end{bmatrix}}_{\reqdef \Gamma} \bs{v} + \underbrace{\begin{bmatrix}
	 	I\\
	 	A\\
	 	A^2\\
	 	\vdots\\
	 	A^{T}
 	\end{bmatrix}}_{\reqdef \Theta} x.
 	$$
Thus, we have
 	$H \eqdef \bar{R} + \Gamma^\top \bar{Q} \Gamma$,
 	$D \eqdef \Gamma^\top \bar{Q} \Theta$,
where $\bar{R} \eqdef I\otimes R$ and $\bar{Q} \eqdef \textrm{diag}(I \otimes Q, P)$. With bounded polyhedral constraints acting both on the state and input, i.e., $\Xi x_i \le \xi$ and $\Upsilon v_i \le \varrho$ for given pairs of matrix/vector $(\Xi, \xi)$ and $(\Upsilon, \varrho)$, and all $i \in \mc{T}$, we finally obtain $N \eqdef \col((I \otimes \Xi) \Gamma, I \otimes \Upsilon)$, $d \eqdef \col(\xi \bsone, \varrho \bsone)$ and $S \eqdef N H^{-1} D^\top - \col((I \otimes \Xi) \Theta, \bs{0})$.

Let $\mc{Z}(x) \eqdef \{\bs{z} \in \R^{mT} \mid N \bs{z} \leq d + S x\}$ be the feasible set of \eqref{eq:QP_final} for any $x \in \mc{X}$. We assume without loss of generality that $\rank(S) = n$, since otherwise the problem can be reduced to an equivalent form by considering a smaller set of parameters~\cite{borrelli2017predictive}.

In order to avoid pathological cases when computing the (unique) solution to the \gls{mp-QP} in \eqref{eq:QP_final}, we will make a further standard assumption about the constraints.   Let $p \geq mT$ be the number of linear constraints in \eqref{eq:QP_final}, and let $\mc{P} \eqdef \{1, \ldots, p\}$ be the associated set of indices. For some $x \in \mc{X}$, define the sets of active constraints at a feasible point $\bs{z} \in \mc{Z}(x)$ of \eqref{eq:QP_final} as:
$$
\begin{aligned}
	\mc{A}(x) &\eqdef \{i \in \mc{P} \mid N_{i,:}  \bs{z} - d_i - S_{i,:}  x = 0\}.
\end{aligned}
$$

\begin{standing}\label{standing:LICQ}\textup{(Linear independence constraint qualification \cite[Def.~2.1]{borrelli2017predictive})}\label{def:LICQ}
	For any $x \in \mc{X}$, the \normaltext{\gls{LICQ}} is said to hold at a feasible point $\bs{z} \in \mc{Z}(x)$ of \eqref{eq:QP_final} if $N_{\mc{A}(x)}$ has linearly independent rows.
   For all $x \in \mc{X}$, the \normaltext{\gls{LICQ}} is assumed to hold for the \normaltext{\gls{mp-QP}} in \eqref{eq:QP_final}.
	\hfill$\square$
\end{standing}

It is known that \gls{LICQ} is sufficient to rule out the possibility that more than $mT$ constraints are active at a given feasible point $\bs{z} \in \mc{Z}(x)$, thereby avoiding primal degeneracy, \cite[\S 4.1.1]{bemporad2002explicit}). It follows that under \gls{LICQ} the problem dual to \eqref{eq:QP_final} is a strictly convex program, and therefore its optimal solution is characterized by a unique vector of Lagrange multipliers.

\begin{remark}
	In stating Standing Assumption~\ref{standing:LICQ} we have implicitly assumed that \eqref{eq:QP_final} is feasible for all $x \in \mc{X}$. If this is not the case, then we may restrict Standing Assumption~\ref{standing:LICQ} to those states that make \eqref{eq:QP_final} feasible (see \normaltext{\cite[Th.~6.1]{borrelli2017predictive}}).
	\hfill$\square$
\end{remark}

For any subset of indices $\mc{A} \subseteq \mc{P}$, we define the \emph{critical region} of states $x$ associated with the set of active constraints $\mc{A}$ as $\mc{R}_{\mc{A}} \eqdef \{x \in \mc{X} \mid \mc{A}(x) = \mc{A}\}$, which results into a polyhedral set \cite[Th.~6.6]{borrelli2017predictive}. Collectively these critical regions represent a valid partition of $\mc{X}$, and each of them has an associated component of $u_\textrm{MPC}(\cdot)$ \cite{bemporad2002explicit}, according to Definition~\ref{def:pwa}. This amounts to the explicit version (i.e., \gls{eMPC}) of the \gls{MPC} policy defined by the optimization problem in \eqref{eq:QP_init}.

\section{Maximum gain computation as a\\ mixed-integer linear program}\label{sec:maximal_gain}
We next develop a method of computing the maximum gain \cite{darup2017maximal} (and hence the Lipschitz constant, according to Lemma~\ref{lemma:pwa_lip}) of the \gls{MPC} policy $u_{\textrm{MPC}}(\cdot)$ directly vian \gls{MI} programming.

Note that the maximum gain can also be computed by means of available tools that compute the complete explicit solution to \gls{mp-QP} in \eqref{eq:QP_final} directly, e.g., the \gls{MPC} Toolbox \cite{bemporad2021model}. However, from \cite{jordan2020exactly} we know that the Lipschitz constant of a \gls{ReLU} network, which we will use to approximate the \gls{MPC} policy in \eqref{eq:QP_init}, can itself be computed through an \gls{MILP}. We therefore require a technique compatible with the one proposed in \cite{jordan2020exactly}, which will also allow us subsequently to compute key quantities characterizing the approximation error $e(\cdot)$, according to \S \ref{sec:stability}.

We first require the following two ancillary results:

\begin{lemma}\label{lemma:MI_matrix_norm}
	For any $K \in \R^{m \times n}$ and $\alpha \in \{1, \infty\}$, the norm $\|K\|_{\alpha}$ can be computed by solving a linear program that admits a binary vector as its optimizer.
	\hfill$\square$
\end{lemma}
\begin{proof}
	Let $\alpha =  1$. By definition, we have $\|K\|_1 = \textrm{max}_{j = 1, \ldots, n} \ \{\bsone^\top |K_{:,j}|\}$, where the absolute value $|\cdot|$ is applied element-wise along the $j$-th column of $K$.  Using standard techniques, we obtain the \gls{LP}
\begin{align}
	\|K\|_1 &= \left\{
	\begin{aligned}
		& \underset{t, (s_j)_{j = 1}^{n}}{\textrm{min}} & & t\\
		& \hspace{2.5mm}\textrm{ s.t. } & & \bsone^\top s_j \leq t, \; j = 1, \ldots, n,\\
		&&& - s_j \leq K_{:,j} \leq s_j, \; j = 1, \ldots, n,
	\end{aligned}
	\right. \label{eq:norm_comp_LP_verbose}
\intertext{or, more compactly,}
		\|K\|_1 &= \left\{
		\begin{aligned}
			& \underset{\eta}{\textrm{min}} & & c^\top \eta\\
			& \hspace{0cm}\textrm{ s.t. } & & M \eta \leq h,\\
		\end{aligned}
		\right. \label{eq:norm_comp_LP}
	\end{align}
	where $s \!\eqdef\! \col((s_j)_{j = 1}^n) \!\in\! \R^{mn}$, $\eta \!\eqdef\! \col(s, t) \!\in\! \R^{mn + 1}$, $c \!\eqdef\! \col(\bs{0},1) \!\in\! \R^{mn + 1}$, $M \!\in\! \R^{n(1 + 2m) \times (mn + 1)}$, $h \!\in\! \R^{n(1 + 2m)}$, with
	\begin{equation}\label{eq:matrices}
			M \eqdef \begin{bmatrix}
			\bsone^\top \otimes I & -\bsone\\
			- I & \phantom{-}\bs{0}\\
			- I & \phantom{-}\bs{0}
		\end{bmatrix} , \text{ and } h \eqdef \begin{bmatrix}
			\phantom{-}\bs{0}\\
			\phantom{-}\vvec(K)\\
			-\vvec(K)
		\end{bmatrix}.
	\end{equation}
	 The associated dual problem is then
	\begin{equation}\label{eq:norm_comp_LP_dual}
		\|K\|_1 = \left\{
		\begin{aligned}
			& \underset{\lambda \geq 0}{\textrm{max}} & & -h^\top \lambda\\
			& \hspace{0cm}\textrm{ s.t. } & & M^\top \lambda = - c,\\
		\end{aligned}
		\right.
	\end{equation}
 and strong duality holds since \eqref{eq:norm_comp_LP_verbose} is always feasible  \cite[\S 5.2]{boyd2004convex}. 

We next show how to construct a binary dual optimizer $\lambda^\star \in \mathbb{B}^{n(1 + 2m)}$ for \eqref{eq:norm_comp_LP_dual}.   Let $\iota \in \{1, \ldots, n\}$ be the (possibly not unique) index associated with a column of $K$ such that $\|K\|_1 = \bsone^\top |K_{:,\iota}|$.   Partition the multiplier into $\lambda^\star \eqdef \col(\lambda^0,\lambda^1,\lambda^2)$, where $\lambda^0 \in \mathbb{B}^n$ and $\lambda^1, \lambda^2 \in \mathbb{B}^{mn}$.  Set the $\iota$-th element of $\lambda^0$ to $1$, with all other elements zero.   Set the multiplier $\lambda^1$ to $1$ at those indices corresponding to elements of $\vvec{(K)}$ that are both in the $\iota$-th column and negative, and zero elsewhere.  Construct the multiplier $\lambda^2$ similarly, but for positive elements of $\vvec{(K)}$.  It is then straightforward to confirm that $M^\top \lambda^\star = - c$ and
$-h^\top \lambda^\star = \|K\|_1$.   Proof of the result $\|K\|_\infty$ is similar.
\end{proof}


\begin{proposition}\label{prop:MILP_matrix_norm}
	Suppose $\mc{X} \subseteq \R^n$ is a polytope and $K : \mc{X} \to \R^{m \times n}$ an affine function. Then computing $\mc{L}_{\alpha}(K,\mc{X}) = \normaltext{\textrm{max}}_{x \in \mc{X}} \, \| K(x) \|_\alpha$ amounts to an \normaltext{\gls{MILP}} for $\alpha \in \{1, \infty\}$.
	\hfill$\square$
\end{proposition}

\begin{proof}

	Since $\mc{X}$ is bounded and $K$ is affine on $\mc{X}$, there exist matrices $\underline{K}$, $\overline{K}$ such that $\underline{K} \le K(x) \le \overline{K}$, where the inequalities apply element-wise. Let $\alpha =  1$.  From \eqref{eq:norm_comp_LP_dual},
	$$
	\mc{L}_{1}(K,\mc{X}) \!=\! \underset{x \in \mc{X}}{\textrm{max}} \, \|K(x)\|_1 \!=\! \left\{
	\begin{aligned}
		& \underset{x, \lambda}{\textrm{max}} & & -h(x)^\top \lambda\\
		& \hspace{0cm}\textrm{ s.t. } & & M^\top \lambda = - c,\\
		&&& x \in \mc{X}, \lambda \in \mathbb{B}^{n(1 + 2m)},
	\end{aligned}
	\right.
	$$
	where we have substituted the constraint $\lambda \geq 0$ with $\lambda \in \mathbb{B}^{n(1 + 2m)}$ due to Lemma~\ref{lemma:MI_matrix_norm} and defined $h(x) = \col(\bs{0}, k(x), -k(x))$ and $k(x) \eqdef \vvec(K(x))$.   Partition the binary variable $\lambda$ as $\lambda = \col(\lambda^0,\lambda^1,\lambda^2)$ with $\lambda^1,\lambda^2 \in \mathbb{B}^{nm}$, so that the objective function becomes $-h(x)^\top \lambda = -k(x)^\top\lambda^1 + k(x)^\top\lambda^2$.   Using standard \gls{MI} modelling techniques (e.g. \cite{bemporad1999control}), one can introduce $y \eqdef \col(y^1,y^2) \in \mathbb{R}^{2mn}$ and appropriate \gls{MI} linear inequalities such that $[\lambda^i_j = 0] \implies [y^i_j = 0]$, and $[\lambda^i_j = 1] \implies [y^i_j = k_j(x)]$ for $i = 1,2$ and $j = 1, \ldots, mn$.

	Rewriting these additional inequalities as $E \, \col(x, \lambda, y) \leq g$ (this is always possible since any element of $k(\cdot)$ is affine and bounded over $\mc{X}$) with some appropriately constructed matrix $E \in \R^{8mn \times (4m+1)n}$ and vector $g \in \R^{8mn}$, we get
	\begin{equation}\label{eq:norm_comp_MILP}
		\mc{L}_{1}(K,\mc{X}) = \left\{
		\begin{aligned}
			& \underset{x, \lambda, y}{\textrm{max}} & & -\bsone^\top y\\
			& \hspace{0cm}\textrm{ s.t. } & & M^\top \lambda = - c,\\
			&&& E \, \col(x, \lambda, y) \leq g,\\
			&&& x \in \mc{X}, \lambda \in \mathbb{B}^{n(1 + 2m)},
		\end{aligned}
		\right.
	\end{equation}
	which is an \gls{MILP}. Proof of the result for $\mc{L}_{\infty}(K,\mc{X}) = \textrm{max}_{x \in \mc{X}} \, \|K(x)\|_{\infty}$ follows similar arguments. \qedhere

\end{proof}


Proposition~\ref{prop:MILP_matrix_norm} says that the norm of a matrix whose entries are affine in $x \in \mc{X}$ can be computed through an \gls{MILP}. We now state and prove the main result of this section, which says that the maximum matrix norm taken over the entire partition induced by $u_\textrm{MPC}(\cdot)$ can also be computed via an \gls{MILP}:

\begin{theorem}\label{th:norm_comp}
	Let $\alpha \in \{1, \infty\}$. Then computing $\mc{L}_{\alpha}(u_\normaltext{\textrm{MPC}},\mc{X})$ amounts to an \normaltext{\gls{MILP}}.
	\hfill$\square$
\end{theorem}

\begin{proof}
Recalling that the \gls{QP} in \eqref{eq:QP_final} is assumed strictly convex and introducing a vector of nonnegative slacks $r \in \R_{\ge 0}^{p}$ and inequality multipliers $\mu \in \R_{\geq 0}^{p}$, for each $x \in \mc{X}$ the KKT conditions for \eqref{eq:QP_final} are
\begin{equation}\label{eq:KKT_1}
	\left\{
	\begin{aligned}
		& \bs{z} + H^{-1} N^\top \mu = 0,\\
		& N \bs{z} + r - d - S x = 0,\\
		& r \geq 0,\\
		& 0 \leq \mu \perp r.
	\end{aligned}
	\right.
\end{equation}

The complementarity condition can be rewritten by introduction of a vector of binary variables $\sigma\in\mathbb{B}^{p}$ such that $[\sigma_i = 1] \implies [r_i = 0]$; otherwise $r$ can take any value within its range $[0, \bar{r}]$ with upper bound $\bar{r} > 0$, guaranteed to exist since the primal modified \gls{QP} in \eqref{eq:QP_final} is assumed feasible.
For $x \in \mc{X}$ the conditions in \eqref{eq:KKT_1} then translate to
\begin{equation}\label{eq:KKT}
	\left\{
	\begin{aligned}
		& \bs{z} + H^{-1} N^\top \mu = 0,\\
		& N \bs{z} + r - d - S x = 0,\\
		& 0 \leq r  \leq \bar{r} (\bsone - \sigma) ,\\
		& 0 \leq \mu \leq \bar{\mu} \sigma.
	\end{aligned}
	\right.
\end{equation}
Assuming the primal modified \gls{QP} in \eqref{eq:QP_final} to be feasible likewise implies the existence of an upper bound for $\mu$, namely $\bar{\mu} > 0$, since $\mc{X}$ amounts to a bounded polyhedral set, hence compact.

Given Standing Assumption~\ref{standing:LICQ}, each critical region of active constraints is uniquely determined by a vector of active constraints (since the dual of the \gls{mp-QP} in \eqref{eq:QP_final} is strictly convex), the indices of which are encoded in $\sigma$.
In addition, solving the system \eqref{eq:KKT} for some $x$ yields the optimal control
\begin{equation}\label{eq:single_control_x}
	u_{\textrm{MPC}}(x) = C(\bs{z} - H^{-1} D^\top x),
\end{equation}
with selection matrix $C \eqdef [I \; 0 \,  \cdots \, 0]$. This corresponds to an affine law $u_{\textrm{MPC}}(x) = K(\sigma) x + c(\sigma)$, for some gain matrix $K(\sigma) \in \R^{m \times n}$ and vector $c(\sigma) \in \R^m$ unique to the particular set of active constraints encoded by $\sigma$, neither of which we have needed to characterize explicitly.

\begin{figure}[t!]
	\centering
	\ifTwoColumn
\else
		\includegraphics[width=.5\columnwidth]{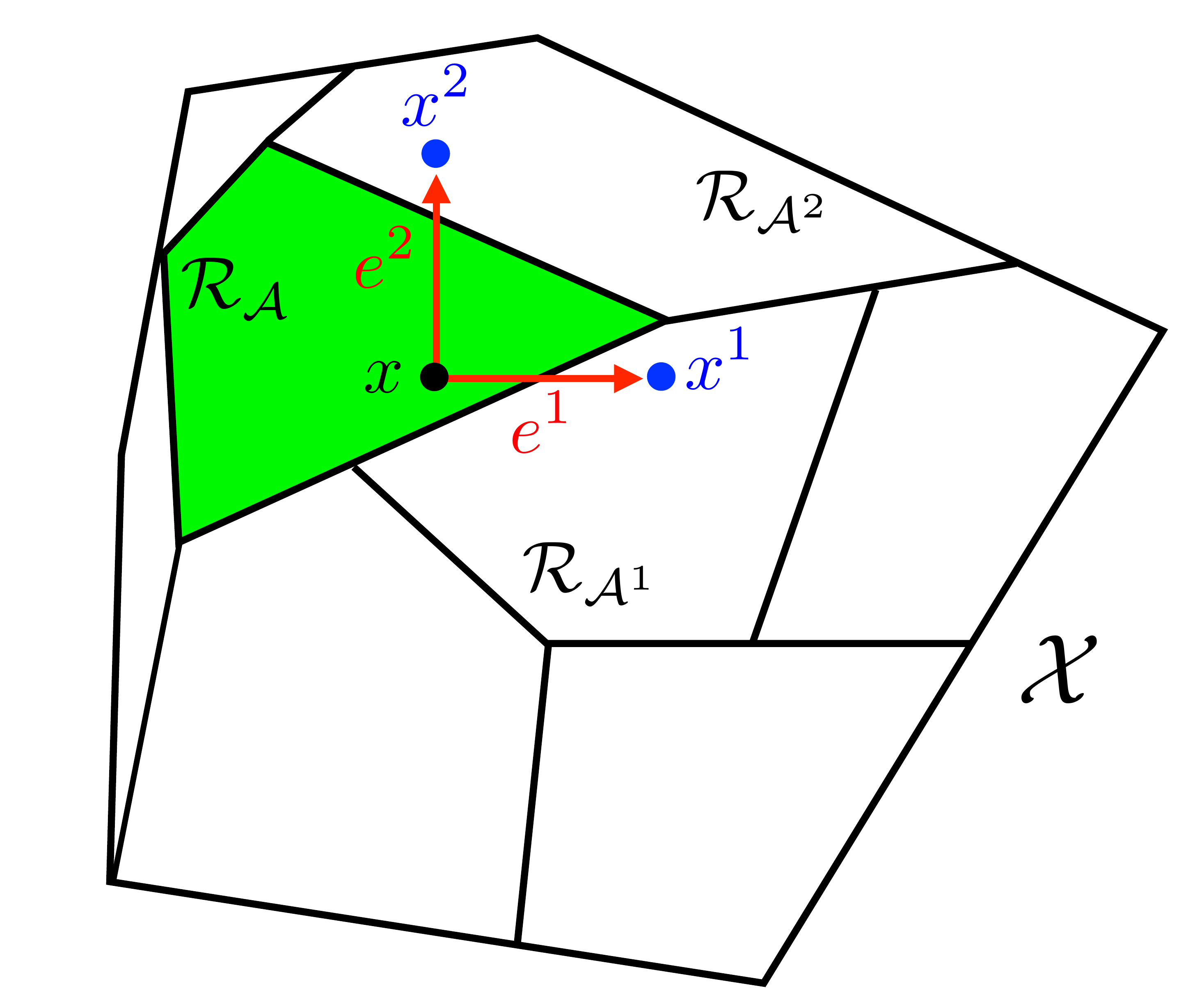}
	\fi
	\caption{Two-dimensional schematic representation of the line of proof of Theorem~\ref{th:norm_comp}.}
	\label{fig:active_region_MPC}
\end{figure}

To compute $\mc{L}_{\alpha}(u_{\textrm{MPC}},\mc{X})$, $\alpha \in \{1, \infty\}$ without explicit calculation of $u_{\textrm{MPC}}(\cdot)$ across all of its partition regions, we first perturb $x$ along the canonical basis  vectors in $\R^n$ and consider how the optimal solution to the \gls{QP} in \eqref{eq:QP_final} varies, provided that the same set of active/inactive constraints is imposed, according to the binary vector $\sigma$ (see Fig.~\ref{fig:active_region_MPC} for an illustration).

We then introduce real auxiliary variables $\{x^i, \, \bs{z}^i, \, \mu^i, \, r^i\}$ for $i = 1, \ldots, n$, and additional \gls{MI} linear constraints as
\begin{equation}\label{eq:KKT_perturbed}
	\forall i \in \{1, \ldots, n\} : \left\{
	\begin{aligned}
		& x^i = x + e^i,\\
		& \bs{z}^i + H^{-1} N^\top \mu^i = 0,\\
		& N \bs{z}^i + r^i - d - S x^i = 0,\\
		& - \bar{r} (\bsone - \sigma) \leq r^i  \leq \bar{r} (\bsone - \sigma),\\
		& - \bar{\mu} \sigma \leq \mu^i \leq \bar{\mu} \sigma,
	\end{aligned}
	\right.
\end{equation}
	where $e^i \in \R^n$ is the $i$-th vector of the canonical basis.  The nonnegativity of both $r^i$ and $\mu^i$ is relaxed to guarantee the existence of a solution to \eqref{eq:KKT_perturbed}, since $x^i$ may fall within a critical region with active constraints that differ from those for $x$, as shown in Fig.~\ref{fig:active_region_MPC}. Note that only the state $x$, which serves as a parameter, is varied to obtain $x^i$, while the newly introduced $\{\bs{z}^i, \, \mu^i, \, r^i\}$ are additional decision variables, subject to the \gls{MI} linear constraints in \eqref{eq:KKT_perturbed}, which allow us to define
	\begin{equation}\label{eq:single_control}
		\forall i \in \{1, \ldots, n\} : u^i = C(\bs{z}^i - H^{-1} D^\top x^i).
	\end{equation}
	Observe that $u^i$ may differ from $u_{\textrm{MPC}}(x^i) = u_{\textrm{MPC}}(x + e^i)$, since the active set encoded by $\sigma$ for some $x$ may differ from the active set at the perturbed point $x + e^i$, particular for those $x$ near the boundary of their partition.   However, it still holds that $u^i = K(\sigma) x^i + c(\sigma)$, and hence that
\begin{align*}
\left[ u^1 \, \cdots \, u^n \right]
&= K(\sigma) \left[ x^1 \, \cdots \, x^n \right] + c(\sigma) \otimes \bsone^\top\\
&= K(\sigma) \left[x \otimes \bsone^\top + \left[ e^1 \, \cdots \, e^n \right]\right] + c(\sigma) \otimes \bsone^\top\\
&= \left(K(\sigma) x + c(\sigma)\right) \otimes \bsone^\top + K(\sigma )
\end{align*}
and we can isolate the gain term $K(\sigma)$ directly to obtain
\begin{equation}\label{eq:controller_expression}
K(\sigma) =  \left[ u^1 \, \cdots \, u^n \right]- u_{\textrm{MPC}}(x)\otimes \bsone^\top.
\end{equation}
The result is that we have constructed an expression for the controller gain in the critical region parametrized by some choice of $\sigma$, which can itself be computed numerically for any $x \in \mc{X}$ by solving the set of \gls{MI} linear constraints \eqref{eq:KKT}--\eqref{eq:single_control}.   Finally, combining all of the additional variables and constraints introduced in \eqref{eq:KKT}--\eqref{eq:controller_expression}, we can apply Proposition~\ref{prop:MILP_matrix_norm} to construct an \gls{MILP} in the spirit of \eqref{eq:norm_comp_MILP} to compute $\mc{L}_{\alpha}(u_{\textrm{MPC}},\mc{X})$, for some given $\alpha \in \{1, \infty\}$.
\end{proof}
In Appendix~\ref{app:maximal_gain_MPC} we look more closely at the procedure described in Theorem~\ref{th:norm_comp} and contrast numerically our proposed approach with the solution obtained via the \gls{MPC} Toolbox \cite{bemporad2021model}.


\section{Quantifying the approximation quality of\\ piecewise-affine neural networks}\label{sec:training_quality}

We can now develop computational results that ensure the stability of a \gls{ReLU}-based control policy $u_{\textrm{NN}}$ constructed based on approximation of a stabilizing \gls{MPC} law $u_{\textrm{MPC}}$.

Since the \gls{MPC} policy $u_{\textrm{MPC}}$ is designed to (exponentially) stabilize the \gls{LTI} system in \eqref{eq:LTI_sys} to the origin, then we may expect that the \gls{ReLU} based policy should also be stabilizing if the approximation error $e(\cdot) = u_\normaltext{\textrm{NN}}(\cdot) - u_\normaltext{\textrm{MPC}}(\cdot)$ is sufficiently small.  This error function is the difference of \gls{PWA} functions, and so also \gls{PWA} \cite[Prop.\ 1.1]{gorokhovik1994pwa}. Thus, it can similarly be shown to be bounded and Lipschitz continuous on $\mc{X}$, and we can therefore apply the results of \S \ref{sec:stability} to find conditions under which stability is preserved.  
We first develop some properties of the approximation error mapping $e(\cdot)$:

\begin{theorem}\label{th:error_comp}
	For $\alpha \in \{1, \infty\}$, the approximation error $e(\cdot) = u_\normaltext{\textrm{NN}}(\cdot) - u_\normaltext{\textrm{MPC}}(\cdot)$ has the following properties:
	\begin{itemize}
		\item[(i)] The maximal error $\normaltext{\textrm{max}}_{x \in \mc{X}} \|e(x)\|_\alpha \eqqcolon \bar{e}_{\alpha}$ can be computed by solving an \normaltext{\gls{MILP}};
		\item[(ii)] The Lipschitz constant $\mc{L}_\alpha(e, \mc{X})$ can be computed by solving an \normaltext{\gls{MILP}}. \hfill$\square$
	\end{itemize}

\end{theorem}
\begin{proof}
	i) From the proof of Theorem~\ref{th:norm_comp} it follows that $u_{\textrm{MPC}}(\cdot)$ can be computed via the linear expression in \eqref{eq:single_control_x}, provided that $x \in \mc{X}$ and $\bs{z}$, along with assorted other auxiliary variables, satisfy the \gls{MI} linear inequalities in \eqref{eq:KKT}.

	Given the recurrence relation in \eqref{eq:RELU_NN}, the output $F(\cdot)$ of a \gls{ReLU} network in general position can likewise be modelled as some combination of variables satisfying a collection of state-dependent \gls{MI} linear inequalities \cite{fischetti2018deep} (this follows also from \cite[Lemma~1 and 2]{jordan2020exactly}). In fact, given a trained \gls{ReLU} network (i.e., for assigned matrices and vectors $\{(W^{j}, b^{j})\}_{j = 0}^L$), for each $j \in \{0, \ldots, L-1\}$, the internal ``state'' of the \gls{NN} in \eqref{eq:RELU_NN} can be rewritten as $x^{j +1} = \Delta^j (W^j x^j + b^j)$, where the diagonal matrix $\Delta^j \eqdef \textrm{diag}((\delta_i^j)_{i = 1}^{n_{j+1}}) \in \mathbb{B}^{n_{j+1} \times n_{j+1}}$ is such that each element satisfies $[\delta_i^j = 1] \iff [W_i^j x^j + b_i^j \geq 0]$, for all $i \in \{1, \ldots, n_{j + 1}\}$. This logical implication translates into a set of \gls{MI} linear constraints, for an arbitrary small tolerance $\varepsilon > 0$ and appropriate lower/upper bounds $\underline{b}^{j} \leq \bar{b}^j$ of $W_i^j x^j + b_i^j$:
	\begin{equation}\label{eq:first_MI}
		\left\{
			\begin{aligned}
				&-\underline{b}^{j} \delta_i^j  \leq W_i^j x^j + b_i^j - \underline{b}^{j} ,\\
				&(\bar{b}^j + \varepsilon) \delta_i^j \geq W_i^j x^j + b_i^j + \varepsilon.
			\end{aligned}
		\right.
	\end{equation}
	In addition, the bilinear product between the binary variable and the continuous term arising in $x^{j +1}$ can be translated into \gls{MI} linear inequalities by introducing a real auxiliary variable $q^j \in \R^{n_{j+1}}$ satisfying $[\delta^j_i = 0] \implies [q^j_i = 0]$, and $[\delta^j_i = 1] \implies [q^j_i = W_i^j x^j + b_i^j]$, for all $i \in \{1, \ldots, n_{j + 1}\}$. Both of these logical implications translate into \gls{MI} inequalities:
	\begin{equation}\label{eq:second_MI}
		\left\{
		\begin{aligned}
			& \underline{b}^{j} \delta^j_i \leq q^j_i\leq  \bar{b}^{j} \delta^j_i,\\
			& -  \bar{b}^{j} (1 - \delta^j_i) \leq q^j_i - W_i^j x^j - b_i^j\leq -\underline{b}^{j} (1- \delta^j_i).
		\end{aligned}
		\right.
	\end{equation}
	We therefore have $x^{j +1} = q^j$, where the auxiliary variable $q^j$, along with the binary variable $\delta_i^j$, is subject to the \gls{MI} linear inequalities \eqref{eq:first_MI}--\eqref{eq:second_MI}, for all $j \in \{0, \ldots, L-1\}$.  For a given input $x \in \mc{X}$ applied to the \gls{NN}, we therefore have $u_{\textrm{NN}}(x) = W^{L} x^{L} + b^{L} = W^{L} q^{L-1} + b^{L}$, i.e.\ the \gls{NN} can be written as an affine combination of a continuous variable (i.e., $q^{L-1}$) that is required to satisfy some \gls{MI} linear constraints.
	For any norm $\alpha \in \{1,\infty\}$, computing the maximum error $\textrm{max}_{x \in \mc{X}} \|e(x)\|_\alpha = \textrm{max}_{x \in \mc{X}} \|u_{\textrm{NN}}(x) - u_{\textrm{MPC}}(x)\|_\alpha$, then amounts to an \gls{MILP}, since a vector norm maximization problem is a special case of  Lemma~\ref{lemma:MI_matrix_norm} and Proposition~\ref{prop:MILP_matrix_norm}.

	\smallskip

	ii) The approximation error $e(\cdot)$ is a \gls{PWA} mapping \cite[Prop.\ 1.1]{gorokhovik1994pwa}, and from Lemma~\ref{lemma:pwa_lip} its Lipschitz constant coincides with
	\[
	\mc{L}_\alpha(e, \mc{X}) = \underset{x \in \mc{X}}{\textrm{max}} \, \|K_\textrm{e}(x)\|_\alpha =\underset{x \in \mc{X}}{\textrm{max}} \, \|K_\textrm{NN}(x) - K_\textrm{MPC}(x)\|_\alpha,
	\]
	where $K_\textrm{NN}(\cdot)$ is the local linear gain of the \gls{ReLU} network \eqref{eq:RELU_NN}, and $K_\textrm{MPC}(\cdot)$ that of the \gls{MPC} policy \eqref{eq:QP_init}. For $\alpha \in \{1, \infty\}$, the claim follows by relying on Proposition~\ref{prop:MILP_matrix_norm} after noting that:
	\begin{itemize}
		\item $K_\textrm{MPC}(\cdot)$ has the locally linear expression in \eqref{eq:controller_expression} subject to state-dependent \gls{MI} linear constraints \eqref{eq:KKT}--\eqref{eq:single_control};
		\item the Jacobian of $F(\cdot)$ over $\mc{X}$, $\partial F(\mc{X})$, can be encoded as an affine combination of both continuous and binary variables subject to \gls{MI} inequalities \cite[Appendix~D]{jordan2020exactly}.
	\end{itemize}
	Although \eqref{eq:controller_expression} does not directly model each entry of the matrix $K_\textrm{MPC}(\cdot)$, the expression for $K_\textrm{MPC}(\cdot)$ is compatible with the one available in the machine learning literature for $K_\textrm{NN}(\cdot)$, as discussed for instance in \cite[Appendix~D]{jordan2020exactly} or in \cite{fischetti2018deep}. Given any $x \in \mc{X}$, we indeed note that
	 applying the chain rule for the derivative to the recurrence  \eqref{eq:RELU_NN} with $x^{j +1} = \Delta^j (W^j x^j + b^j)$, $j \in \{0, \ldots, L-1\}$, leads to
	$$
		K_\textrm{NN}(x) = W^L \Delta^{L-1} W^{L-1} \cdots \Delta^{0} W^{0},
	$$
	where the dependence on $x$ comes via the diagonal matrix $\Delta^0$,  and in cascade by any $\Delta^j$, whose elements are subject to \eqref{eq:first_MI}. Define auxiliary matrices $Y^{j+1} \eqdef W^{j+1} \Delta^j Y^j \in \R^{n_{j+2} \times n}$, for all $j \in \{0, \ldots, L-1\}$, with $Y^0 \eqdef W^0 \in \R^{n_1 \times n}$.  Since any row of the matrix $\Delta^j W^j = \col((\delta^j_i W^{j}_i)_{i = 1}^{n_{j+1}}) \eqqcolon Z^j \in \R^{n_{j+1} \times n_j}$ coincides with the associated row of $W^{j}$ only if $\delta^{j}_i = 1$ ($0$ otherwise), any entry of $Z^j$ satisfies the set of \gls{MI} linear inequalities, for bounds $\underline{w}_{h,k}^j \leq \bar{w}_{h,k}^j$ on $w_{h,k}^j$
	\begin{equation}\label{eq:third_MI}
		\left\{
			\begin{aligned}
				& \underline{w}_{h,k}^j \delta^j_h \leq z^j_{h,k}\leq  \bar{w}_{h,k}^j \delta^j_h,\\
				& -  \bar{w}_{h,k}^j (1 - \delta^j_h) \leq z^j_{h,k}- w_{h,k}^j,\leq -\underline{w}_{h,k}^j (1- \delta^j_h),
			\end{aligned}
		\right.
	\end{equation}
	it turns out that $K_\textrm{NN}(x) = Y^{L}$, subject to \eqref{eq:first_MI} and \eqref{eq:third_MI} for all $j \in \{0, \ldots, L-1\}$.
	The local linear gain $K_\textrm{NN}(\cdot)$ of a \gls{ReLU} network in general position can hence be computed via Proposition~\ref{prop:MILP_matrix_norm}, and this concludes the proof.
\end{proof}


Theorem~\ref{th:error_comp} provides an offline, optimization-based procedure to compute exactly both the worst-case approximation error between the \gls{PWA} mappings associated with the \gls{ReLU} network in \eqref{eq:RELU_NN} and the \gls{MPC} law in \eqref{eq:QP_init}, as $\|e(x)\|_\alpha  \leq \bar{e}_{\alpha}$, for all $x \in \mc{X}$, and the associated Lipschitz constant over $\mc{X}$, $\mc{L}_\alpha(e, \mc{X})$, for $\alpha \in \{1, \infty\}$.   These quantities are precisely of the type required to apply the stability results of \S \ref{sec:stability}.

Note that Lemma~\ref{lemma:NN_stability_local} and Theorem~\ref{th:exp_conv} require one to evaluate $\mc{L}_\alpha(e, \mc{X}_\infty)$, in contrast with Theorem~\ref{th:error_comp}.(ii) which provides a method to compute $\mc{L}_\alpha(e, \mc{X})$ (or $\mc{L}_\alpha(e, \Omega_b)$ in view of Remark~\ref{remark:omega_b}).  However, $\mc{L}_\alpha(e, \mc{X}_\infty)$ can still be computed by means of the same optimization-based procedure, replacing $\mc{X}$ with $\mc{X}_\infty$ everywhere.  It is known that the polytopic set $\mc{X}_\infty$ can be computed exactly via, e.g., the procedure in \cite{gilbert1991linear}.

Since $K_\textrm{MPC}(x) = \bar{K}_\textrm{MPC}$ for all $x \in \mc{X}_\infty$, $\mc{L}_\alpha(e, \mc{X}_\infty)$ can also be obtained by solving the \gls{MILP} $\textrm{max}_{x \in \mc{X}_\infty} \, \|K_\textrm{NN}(x) - \bar{K}_\textrm{MPC}\|_\alpha$, where the unconstrained optimal gain $\bar{K}_\textrm{MPC}$ can be computed offline through a least-square approach  \cite[\S 6.1.1]{rawlings2017model}.
In both cases, however, at least an estimate of the set $\mc{X}_\infty$ is required.

Alternatively, one could simply employ $\mc{L}_\alpha(e, \mc{X})$ directly in \eqref{eq:upper_bound_lipschitz} in place of $\mc{L}_\alpha(e, \mc{X}_\infty)$, since $\mc{X}_\infty \subseteq \mc{X}$ implies $\mc{L}_\alpha(e, \mc{X}_\infty) \leq \mc{L}_\alpha(e, \mc{X})$.  This  comes at the cost of greater conservatism however, potentially leading to design a \gls{ReLU}-based controller with greater complexity than is required.


\section{Discussion of reliably-stabilizing\\ \gls{PWA-NN} controllers}\label{sec:discussion}

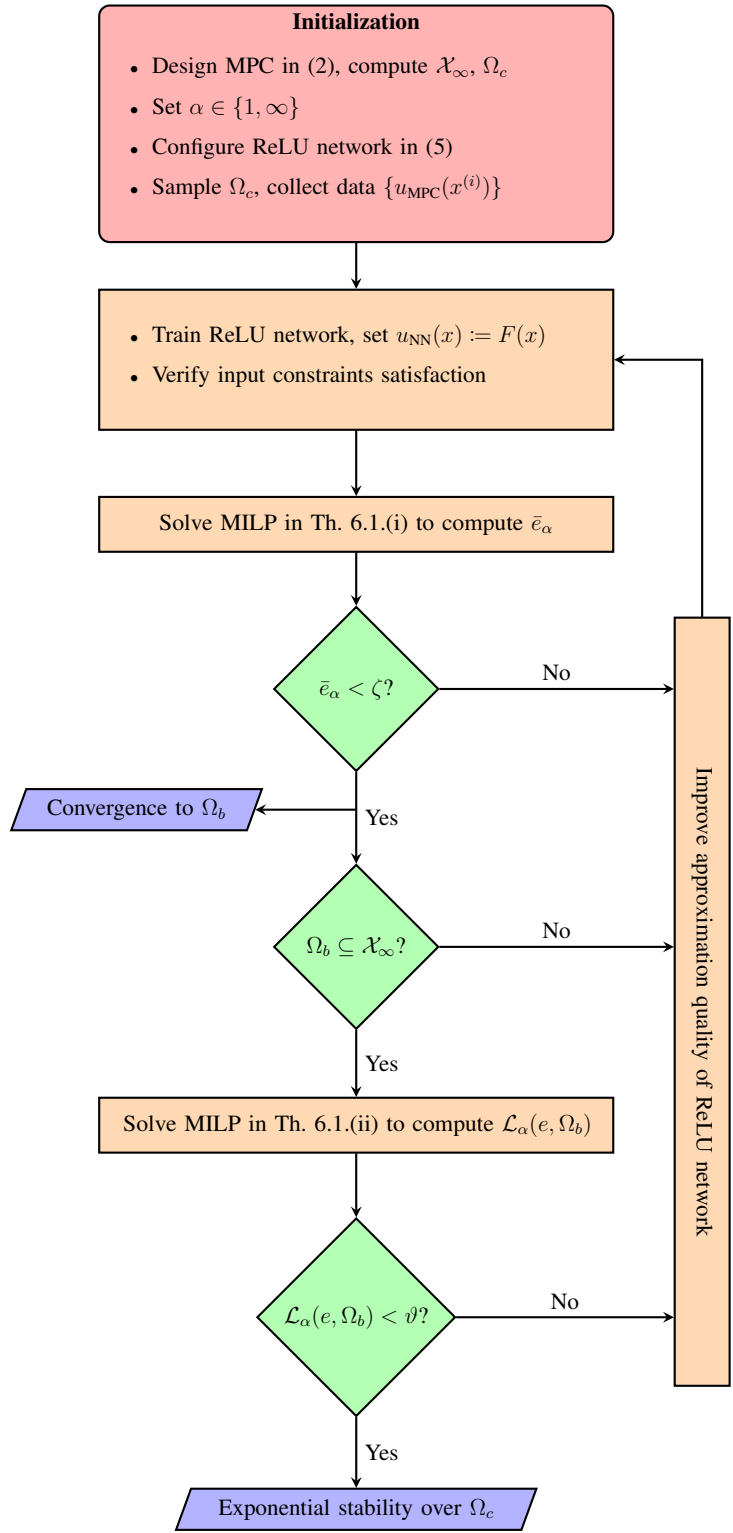
\begin{figure}[t!]
	\centering
	\ifTwoColumn
\else
		\tikzstyle{startstop} = [rectangle, rounded corners, minimum width=2.5cm, text width=.55\columnwidth,  minimum height=1cm, text centered, draw=black, fill=red!30]
		\tikzstyle{process} = [rectangle, minimum width=1cm, minimum height=1cm, text width=.55\columnwidth, text centered, draw=black, fill=orange!30]
			\begin{tikzpicture}[thick, every node/.style={scale=0.73}, scale=0.73, node distance=2.5cm]
			\node (init) [startstop] { \textbf{Initialization}
					\begin{itemize}
						\item Design \gls{MPC} in \eqref{eq:QP_init}, compute $\mc{X}_\infty$, $\Omega_c$
						\item Set $\alpha \in \{1, \infty\}$
						\item Configure \gls{ReLU} network in \eqref{eq:RELU_NN}
						\item Sample $\Omega_c$, collect data $\{u_\textrm{MPC}(x^{(i)})\}$
						\item[]
					\end{itemize}
			};
			\node (pro1) [process, below of=init, yshift=-1.8cm] {

				\begin{itemize}
					\item Train \gls{ReLU} network, set $u_{\textrm{NN}}(x) \eqdef F(x)$
					\item Verify input constraints satisfaction
					\item[]
				\end{itemize}
				};
			\node (pro2) [process, below of=pro1, yshift=-.5cm] {Solve \gls{MILP} in Th.~\ref{th:error_comp}.(i) to compute $\bar{e}_\alpha$};
			\node (dec1) [decision, below of=pro2, yshift=-.5cm] {$\bar{e}_\alpha < \zeta$?};
			\node (improve) [rectangle, text centered, draw=black, fill=orange!30, minimum height = 14cm, minimum width = 1cm, right of=dec1, xshift=3.8cm, yshift = -5.7cm] {\rotatebox{-90}{Improve approximation quality of \gls{ReLU} network}};
			\node (k) [inner sep=0,minimum size=0,below of=dec1, yshift=.3cm]{};
			\node (dec2) [decision, below of=k, yshift=-.cm] {$\Omega_b \subseteq \mc{X}_\infty$?};
			\node (pro2b) [process, below of=dec2, yshift=-0.75cm] {Solve \gls{MILP} in Th.~\ref{th:error_comp}.(ii) to compute $\mc{L}_\alpha(e, \Omega_b)$};
			\node (dec3) [decision, below of=pro2b, yshift=-1cm] {$\mc{L}_\alpha(e, \Omega_b) < \vartheta$?};
			\node (exp_stab) [io, below of=dec3, yshift=-1cm] {Exponential stability over $\Omega_c$};

			\node (conv) [io, left of=k, xshift=-1.5cm]{Convergence to $\Omega_b$};

			\draw [arrow] (init) -- (pro1);
			\draw [arrow] (pro1) -- (pro2);
			\draw [arrow] (pro2) -- (dec1);
			\draw [arrow] (dec1) -- node[anchor=west] {Yes} (dec2);
			\draw [arrow] (dec2) -- node[anchor=west] {Yes} (pro2b);
			\draw [arrow] (pro2b) -- (dec3);
			\draw [arrow] (dec3) -- node[anchor=west] {Yes} (exp_stab);
			\draw [arrow] (k) -- (conv);
			\draw[arrow] (dec1) -- node[anchor=south] {No} (dec1 -| improve.west);
			\draw[arrow] (dec2) -- node[anchor=south] {No} (dec2 -| improve.west);
			\draw[arrow] (dec3) -- node[anchor=south] {No} (dec3 -| improve.west);
			\draw[arrow] (improve.north) |-  (pro1.east);
		\end{tikzpicture}
	\fi
	\vspace{.4cm}
	\caption{Roadmap for using the proposed results.}
	\label{fig:flowchart}
\end{figure}

Having established that a \gls{ReLU}-based approximation of an \gls{MPC} law guarantees stability if the optimal value of two \glspl{MILP} satisfy certain conditions, we now make some observations and practical suggestions on how to use our results.

\subsection{A user's guide for \gls{PWA-NN} controllers}
Figure~\ref{fig:flowchart} shows a roadmap describing a sequence of decisions involving the results developed in this paper for approximating an \gls{MPC} policy with a \gls{ReLU} network of reasonably low complexity while preserving stability. Specifically, the network complexity is characterized by its \emph{depth} $L$ (the number of hidden layers) and \emph{width} $N$ (the number of neurons).

The initialization step designs an \gls{MPC} controller in \eqref{eq:QP_init} satisfying the conditions of Standing Assumption~\ref{standing:mp-QP}, and computes related quantities. In addition, one has to fix the structure of a \gls{ReLU} network for training collect a certain dataset of samples, $\{u_\textrm{MPC}(x^{(i)})\}$, which is then used to train the \gls{ReLU} network. Note that the data collection phase can easily be done offline by solving the \gls{mp-QP} in \eqref{eq:QP_init} for a collection of state samples $x^{(i)}$ obtained by, e.g., uniformly sampling or gridding $\Omega_c$. Before solving the \gls{MILP} obtained by combining Theorem~\ref{th:error_comp}.(i) and Proposition~\ref{prop:MILP_matrix_norm}, it is crucial to verify whether the resulting controller $u_\textrm{NN}(\cdot)$ is able to generate safe inputs satisfying the constraints (further discussion follows in \S \ref{subsec:input_cons}).
Thus, if the worst-case approximation error meets the condition in \eqref{eq:error_bound} for some value of $\rho$ (see Lemma~\ref{lemma:NN_stability} and related proof), then the \gls{PWA-NN} controller $u_\textrm{NN}(\cdot)$ guarantees that the \gls{LTI} system in \eqref{eq:LTI_sys} is \gls{ISS} and converges to some neighbourhood of the origin $\Omega_b \subset \Omega_c$ exponentially fast. Otherwise, one needs to improve the \gls{ReLU}-based approximation of the \gls{MPC} law.

Due to the many types of \glspl{NN} available, it is challenging to devise a rigorous procedure for improving the approximation quality that holds in general. For this reason, one can find in the literature a variety of empirical recommendations that are not specific to a given type of \gls{NN} or predictive modelling problem, e.g.\ as in \cite{bengio2012practical}.
As a general guidelines, it has been observed in practice that one can achieve better approximation through some combination of increasing the pool of sample points and increasing the complexity of the \gls{ReLU} network structure (in particular, its width) with the same size for all layers. Preparing data prior to modelling by, e.g., standardizing and removing correlations, has also been shown to be beneficial, as well as adopting regularization terms while training the underlying \gls{ReLU} \gls{NN}. 
Note that as a by-product of the \normaltext{\gls{MILP}} in Theorem~\ref{th:error_comp}.(i), one obtains the state associated to the computed worst-case approximation error. A reasonable choice is hence to include that sample upon re-training the~network. 
Since our methodology provides a way to asses the training quality of a given \gls{ReLU} network in replicating the control action of an \gls{MPC} policy, the message conveyed here is that, in case the condition in \eqref{eq:error_bound} is not met, one has to make the worst-case approximation error smaller by implementing a strategy to improve the approximation quality of the \gls{ReLU} network. This will also necessarily require one to re-train the network, and eventually verify input constraints satisfaction.

Then, to design a \gls{PWA-NN} controller that also guarantees exponential convergence to the origin, one has to first look for some pair $(b,\rho)$ such that $b \leq c$ and $\rho \in (\bar{\rho},1)$ for which the inclusion $\Omega_b \subseteq \mc{X}_\infty$ holds, according to Theorem~\ref{th:exp_conv}. Here, $\bar{\rho}$ amounts to the smallest value of $\rho$ for which the condition \eqref{eq:error_bound} is met.
Finally, by solving the \gls{MILP} described in Theorem~\ref{th:error_comp}.(ii), if the condition in Lemma~\ref{lemma:NN_stability_local} is met (possibly with $\mc{L}_\alpha(e, \Omega_b)$ in place of $\mc{L}_\alpha(e, \mc{X}_\infty)$), then $u_\textrm{NN}(\cdot)$ exponentially stabilizes the \gls{LTI} system in \eqref{eq:LTI_sys}. Otherwise, a tailored procedure for improving  the \gls{ReLU}-based approximation of the \gls{MPC} law must be adopted, and the overall process repeated.

\subsection{Accommodating input constraints}\label{subsec:input_cons}
While $u_{\normaltext{\textrm{NN}}}(\cdot)$ guarantees state constraint satisfaction for any initial state $x(0) \in \Omega_c$, the input constraints may not be satisfied. This issue can be rectified in several ways either \emph{before} or \emph{after} applying our methodology and without affecting the proposed results, since they hold for any trained \gls{ReLU} network no matter how the (post-)training is actually performed.

In view of the discussion in \S \ref{sec:stability}, we note that the design of some approximate \gls{MPC} law satisfying input constraints can even be enforced during the design of the original \gls{MPC} scheme in \eqref{eq:QP_init} by focusing on the augmented state variable $\hat{x} \eqdef \col(x, \bar{u})$. This latter evolves according to the dynamics
\begin{equation}\label{eq:augmented_dyn}
	\hat{x}^+ =
	\begin{bmatrix}
		A & 0\\
		0 & 0
	\end{bmatrix} \hat{x} + \begin{bmatrix}
		B\\
		I
	\end{bmatrix} u \, ,
\end{equation}
and hence enables us to incorporate input constraints as state ones directly, since $\hat{x} \in \mc{X} \times \mc{U}$. With \eqref{eq:augmented_dyn} in place of \eqref{eq:LTI_sys}, replicating \emph{mutatis mutandis} the discussion in \S \ref{sec:stability}  allows one to establish the robust positively invariance of some set $\hat{\Omega}_c \subseteq \mc{X} \times \mc{U}$ for the underlying perturbed dynamics, and hence for \eqref{eq:augmented_dyn} with approximated \normaltext{\gls{MPC}} policy $u_{\normaltext{\textrm{NN}}}(\cdot)$. Thus, for any initial state $\hat{x}(0) \in \hat{\Omega}_c$ the resulting trajectory will satisfy both state and input constraints for all $k \in \N$, at the price of introducing some conservativism on the bounds $\zeta$ and $\vartheta$ characterizing worst-case error and Lipschitz constant.

Approaches to enforce input constraints directly during the training phase
of a \gls{NN} are also available in the literature. For example, \cite{markolf2021polytopic} proposes a way of determining the weights of a \gls{NN} that guarantee the satisfaction of input constraints, while \cite{ICLR16-hausknecht} describes a reinforcement learning approach that manipulates the gradient of $u_\textrm{NN}(\cdot)$ with respect to the network parameters as the output nears constraint violation at some sample $x^{(i)} \in \Omega_c$.
Another possibility is to study the reachable set of a trained \gls{NN} through output verification techniques \cite{NEURIPS2018_be53d253,9301422,karg2020stability}.
It has been shown, for example, that the satisfaction of polytopic constraints involving the output of a \gls{NN} can be ensured a-priori by certain properties of common activation functions. Among them, \glspl{ReLU} were used in \cite{9301422} for this purpose, thus requiring one to solve a convex program to check whether or not a \gls{NN} output falls within a desired set.

A further possible approach has been explored in \cite{chen2018approximating,karg2020efficient,paulson2020approximate}, where the output of a trained \gls{ReLU} network is systematically projected through a Dykstra’s projection algorithm onto the polytopic set of feasible control actions parametrized by the current state.
Following this idea, one could even drop the verification of the input constraints satisfaction in the first step of the flowchart in Fig.~\ref{fig:flowchart}, and \emph{after} having verified conditions \eqref{eq:error_bound} and \eqref{eq:upper_bound_lipschitz}, implement $u = \textrm{proj}_{\mc U} (u_{\normaltext{\textrm{NN}}}(x))$ directly to stabilize \eqref{eq:LTI_sys}
In fact, we note that the stability analysis involving the perturbed system \eqref{eq:perturbed_dyn} holds for any state-dependent disturbance $e(x) = u_{\normaltext{\textrm{NN}}}(x) - u_{\normaltext{\textrm{MPC}}}(x)$ bounded in some norm $\alpha \in \{1, \infty\}$ by $\bar{e}_\alpha$. Since the projection mapping is (firmly) nonexpansive \cite[Cor.~12.20]{rockafellar2009variational} and that $u_{\normaltext{\textrm{MPC}}}(x) = \textrm{proj}_{\mc U} (u_{\normaltext{\textrm{MPC}}}(x))$, we have
$$
\begin{aligned}
	&\textrm{max}_{x \in \mc{X}}  \, \|u_{\normaltext{\textrm{NN}}}(x) - u_{\normaltext{\textrm{MPC}}}(x)\|_\alpha \le \bar{e}_\alpha,\\
	& \implies \|\textrm{proj}_{\mc U}(u_{\normaltext{\textrm{NN}}}(x) - u_{\normaltext{\textrm{MPC}}}(x))\|_\alpha \le \bar{e}_\alpha, \text{ for all } x \in \mc{X},\\
	& \implies \|\textrm{proj}_{\mc U}(u_{\normaltext{\textrm{NN}}}(x)) - u_{\normaltext{\textrm{MPC}}}(x)\|_\alpha \le \bar{e}_\alpha, \text{ for all } x \in \mc{X}.
\end{aligned}
$$
Therefore, the theory developed in \S \ref{sec:stability} supports the safe implementation of $u = \textrm{proj}_{\mc U} (u_{\normaltext{\textrm{NN}}}(x))$ while guaranteeing the stabilization of the considered \gls{LTI} system. Similar arguments can be adopted also to show that having verified the condition in Lemma~\ref{lemma:NN_stability_local} implies that the same condition is satisfied when $u_{\normaltext{\textrm{NN}}}(\cdot)$ is replaced by $\textrm{proj}_{\mc U} (u_{\normaltext{\textrm{NN}}}(\cdot))$. We finally remark that in case $\mc{U}$ identifies box constraints, as very frequently happens in practise, $\textrm{proj}_{\mc U}(\cdot)$ simply reduces to a saturation.

\section{Numerical simulations}\label{sec:simulations}
\begin{figure*}[!t]
	\centering
	\ifTwoColumn
\else
		\includegraphics[trim=0cm 11cm 0cm 9cm,clip,width=\columnwidth]{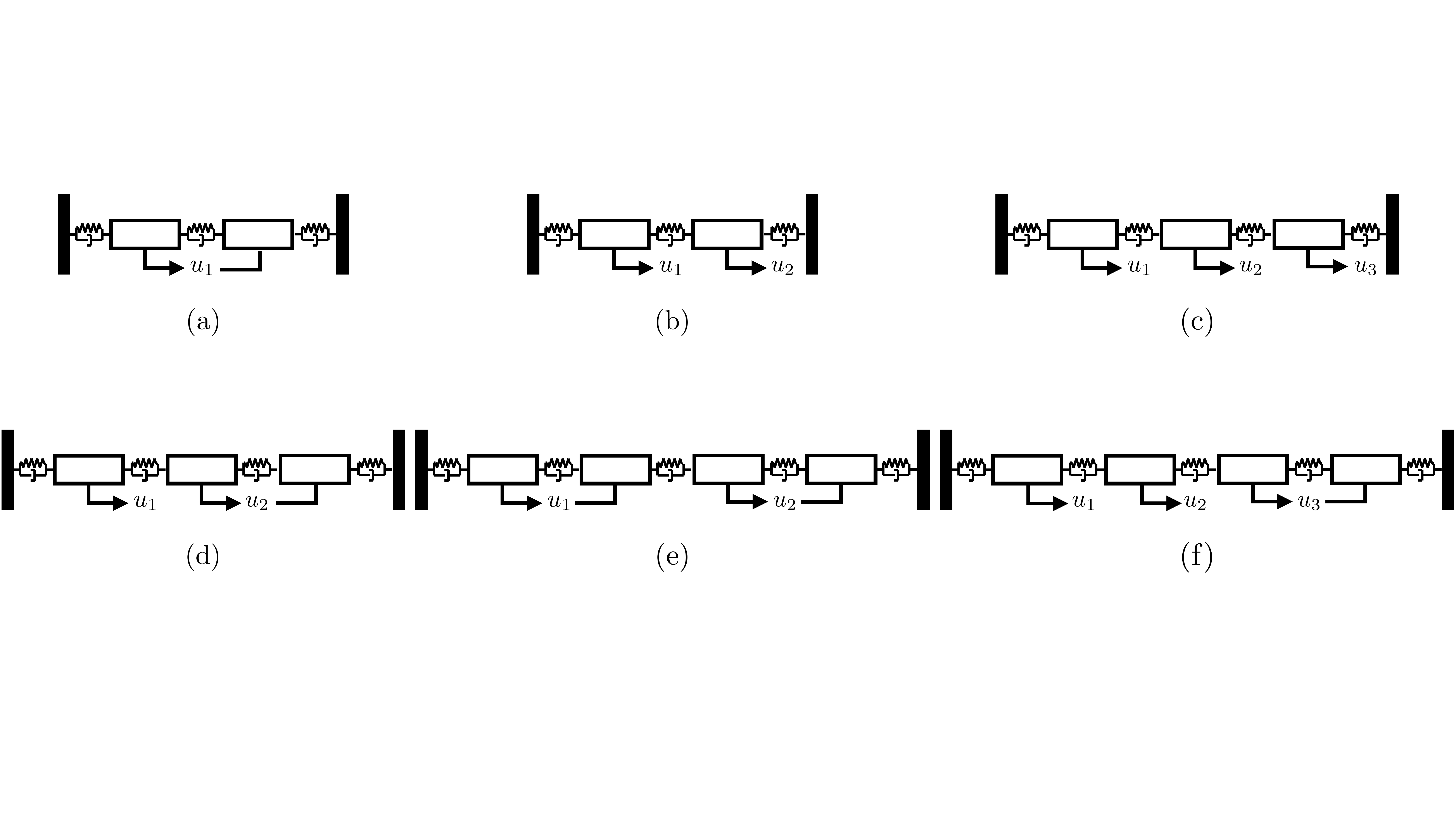}
	\fi
	\caption{Systems of 2, 3 and 4 oscillating masses each with one degree of freedom, connected each other through pairs of spring-damper blocks, and to walls (dark blocks on the sides). Control inputs $u_1$, $u_2$ and $u_3$ are either acting on each single mass at a time, or produce a joint action on multiple masses, according to the direction indicated by each arrow.}
	\label{fig:masses}
\end{figure*}

We now illustrate how to apply our proposed certificates to design minimum complexity \gls{PWA-NN} controllers for stabilizing a system of oscillating masses, each with one degree of freedom as in \cite{wang2009fast,zeilinger2011real,karg2020efficient}. We consider the configurations shown in Fig.~\ref{fig:masses}, thus dealing with systems characterized by four to eight state variables, under the action of up to three control inputs. All simulations were run in Matlab using Gurobi \cite{gurobi} as an \gls{MILP} solver on a laptop with a Quad-Core Intel i5 2.4 GHz CPU and 8 Gb RAM. The \gls{ReLU} networks are all trained by adopting a Levenberg–Marquardt algorithm with mean squared normalized error as a performance function.

As in \cite{zeilinger2011real}, all masses are 1, springs constants 1 and damping constants 0.5. After discretizing the dynamics with sampling rate 0.1, the state $x_i \in \R^2$ of each mass $i \in \{1, \ldots, 4\}$ containing position and velocity is subject to element-wise constraints $|x_i| \leq \col(4, 10)$, while for the control input $u_j$ we have $|u_j| \le 1$, $j \in \{1, 2, 3\}$. We design the \gls{MPC} scheme in \eqref{eq:QP_init} by setting the prediction horizon to $T = 5$, while the weight matrices $Q$ and $R$ are taken as identity matrices. To meet Standing Assumption~\ref{standing:mp-QP}.(i) instead, we choose the terminal weight $P$ and gain $\bar{K}_\textrm{MPC}$ from the unconstrained \gls{LQR} optimal solution. The maximal output admissible set $\mc{X}_\infty$ is computed through the procedure in \cite{gilbert1991linear}, $c_2$ in \eqref{eq:ISS_Lyap_2} as described in \cite[Appendix]{scokaert1997discrete}, $\varsigma$ and $\lambda$ characterizing the upper bound in \eqref{eq:upper_bound_lipschitz} follow from the Gelfand formula \cite[Cor.~5.6.14]{horn2012matrix}, while quantities $c$, $\Omega_c$ and $\mc{L}_2(V_T, \Omega_c)$ have been estimated numerically. For all the considered configurations, the upper bound $\zeta$ is obtained from \eqref{eq:error_bound} with $\rho = 0.9999$.

\begin{table*}[!t]
\caption{Numerical results for the systems of coupled oscillators described in Fig.~\ref{fig:masses}}\label{tab:masses_results}
\centering
	\begin{tabular}{ccccccccccccc}
		\toprule
		\multirow{2}[2]{*}{Case} & \multicolumn{2}{c}{\gls{eMPC}} & & \multicolumn{9}{c}{\gls{PWA-NN} controller} \\
		\cmidrule{2-3} \cmidrule{5-13} & \multicolumn{1}{c}{\footnotesize{\# of $\mc{R}_{\mc{A}}$}} & \multicolumn{1}{c}{\footnotesize{CPU time [s]}}  &  & \multicolumn{1}{c}{\footnotesize{$N$}} & \multicolumn{1}{c}{\footnotesize{$L$}} & \multicolumn{1}{c}{\footnotesize{$\bar{e}_\infty$}} & \multicolumn{1}{c}{\footnotesize{$\zeta$}} & \multicolumn{1}{c}{\footnotesize{CPU time [s]}} & \multicolumn{1}{c}{\footnotesize{$\mc{L}_\infty(e, \Omega_b)$}} & \multicolumn{1}{c}{\footnotesize{$\vartheta$}} & \multicolumn{1}{c}{\footnotesize{CPU time [s]}} & \multicolumn{1}{c}{\footnotesize{$\textrm{avg}(\mathscr{D})$}}\\
		\midrule
		(a) & 2235 & 121.9  && 41 & 2 & 0.002 & 1.52 & 13.25 & 0.17 & 0.19& 8.35 & $\sim$10\textsuperscript{-6}\\
		\midrule
		(b) & 9161 & 809.8 && 42 & 2 & 0.007 & 1.78 & 224.6 & 0.02 & 0.21 & 0.8 & $\sim$10\textsuperscript{-7}\\
		\midrule
		(c) & * & > 3600  && 73 & 2 & 0.283 & 1 & 1179 & 1.49 & 0.14 & 278.1 & $\sim$10\textsuperscript{-5}\\
		\midrule
		(d) & * & > 3600 && 62 & 2 & 0.009 & 0.72 & 825.4 & 0.04 & 0.12 & 42.81 & $\sim$10\textsuperscript{-7}\\
		\midrule
		(e) & * & > 3600 && 62 & 2 & 0.007 & 0.71 & 1083 & 0.001 & 0.12 & 36.3 & $\sim$10\textsuperscript{-4}\\
		\midrule
		(f) & * & > 3600  && 63 & 3 & 0.031 & 0.47 & 2126 & 0.8 & 0.1 & 1371 & $\sim$10\textsuperscript{-5}\\
		\bottomrule
	\end{tabular}
\end{table*}

The numerical results obtained are summarized in Table~\ref{tab:masses_results} where we have considered the case $\alpha = \infty$, trained each \gls{ReLU} network with $20 \times 10^3$ samples, while the $N-m$ neurons are equally distributed across the $L$ hidden layers. As expected, the control approach based on the optimal explicit solution of \eqref{eq:QP_init} provided by the \gls{MPC} Toolbox \cite{bemporad2021model} is not viable as the dimension of the considered problem grows, while the \gls{PWA-NN} controller based on a \gls{ReLU} approximation of $u_\textrm{MPC}(\cdot)$ still makes possible the stabilization of the systems of coupled oscillators with reasonable offline computation (even far less in those cases admitting a direct comparison, i.e., (a) and (b)).

In fact, the columns referring to \gls{eMPC} show that we can obtain an optimal explicit solution in less than 3600 [s] only for the configurations (a) and (b), while in the remaining cases the simulation was aborted after one hour. In these two scenarios, we also observe that the number of linear regions more than quadruples just by introducing an additional control input acting on the second mass. 
On the other hand, certifying the (exponential) stability guarantees of the minimum complexity \gls{ReLU}-based approximation is still feasible in all configurations considered. As highlighted in Remark~\ref{remark:omega_b}, since the worst-case error $\bar{e}_\infty$ on $\Omega_c$ can be made very small by checking the condition on the Lipschitz constant in \eqref{eq:upper_bound_lipschitz} over $\Omega_b$ rather than the whole of $\mc{X}_\infty$ is preferable, since it allows us to recover exponential stability in all configurations but (c) and (f), while keeping the computational time relatively low (compared to the one for computing $\bar{e}_\infty$). From our numerical experience, indeed, we noticed that satisfying the condition \eqref{eq:upper_bound_lipschitz} with $\mc{L}_\infty(e,\mc{X}_\infty)$ is not trivial, also possibly requiring a much larger amount of time since $\mc{X}_\infty \supset \Omega_b$: this is confirmed by the numerical results obtained for cases (c) and (f).

In the last column of Table~\ref{tab:masses_results} we evaluate the deterioration of control performance of the \gls{PWA-NN} controller relative to the \gls{MPC} policy according to the metric proposed in \cite{zeilinger2011real}. Specifically, we focus on the difference between the cost of the closed-loop trajectory using the optimal control input $u_\textrm{MPC}(\cdot)$, indicated by $x_{\textrm{MPC}}(k)$, $k \in \N$, and the one using the \gls{ReLU}-based controller $u_\textrm{NN}(\cdot)$, i.e., $x_{\textrm{NN}}(k)$, resulting in
\ifTwoColumn
\else
$$
	\mathscr{D}(x(0)) \eqdef \frac{\sum_{k = 0}^\infty  \ell(x_{\textrm{NN}}(k), u_\textrm{NN}(x_{\textrm{NN}}(k))) - \ell(x_{\textrm{MPC}}(k), u_\textrm{MPC}(x_{\textrm{MPC}}(k)))}{\sum_{k = 0}^\infty \ell(x_{\textrm{MPC}}(k), u_\textrm{MPC}(x_{\textrm{MPC}}(k)))},
$$
\fi
where $\ell(x, u) \coloneqq \tfrac{1}{2} (\|x\|^2_Q + \|v\|^2_R)$ is the stage cost in \eqref{eq:QP_init}, and $x_{\textrm{MPC}}(0) = x_{\textrm{NN}}(0) = x(0)$.
By considering the average closed-loop performance deterioration taken over $10^3$ initial conditions uniformly sampled in $\Omega_c$ with $\textrm{max} \{\|x_{\textrm{MPC}}(k)\|_2, \|x_{\textrm{NN}}(k)\|_2\} \le 10^{-3}$ as stopping criterion, we observe no substantial performance degradation with nearly coincident close-loop trajectories. Note that these considerations hold even for those cases where the approximation quality of the \gls{PWA-NN} controller is not good enough to meet the condition involving $\mc{L}_\infty(e,\Omega_b)$, i.e., configurations (c) and (f). This was expected as the conditions of \S \ref{sec:stability} are sufficient only, hence conservative.

A common drawback of \gls{MI} optimization is poor scalability with increasing problem size.   However, we observe in
Table~\ref{tab:masses_results} that, fixing the number of inputs, the computation time is only weakly dependent on the state dimension $n$ -- compare for instance scenarios (d) and (e), and eventually also (b) (though this latter considers a different number of neurons).   On the other hand, the numerical results indicate that computation time is very sensitive to the number of inputs $m$, as is evident by contrasting configurations (a)--(b), and (e)--(f) separately (i.e., fixing the state dimension).   For larger problems the non-negligible offline computational efforts exhibited by our certification  method with a limited number of neurons and layers provide a further motivation to design minimum complexity \gls{ReLU}-based controllers that can be implemented on dedicated hardware up to tens of MHz \cite{duarte2018fast,zhang2019real,schindler2020real}.

\section{Conclusion and Outlook}

The implementation of controllers that closely approximate the action of \gls{MPC} policies with minimal online computational load is a critical consideration for fast embedded systems.
We have shown that the design of  \gls{ReLU}-based approximations with provable stability guarantees require one to construct and solve two \glspl{MILP} offline, whose associated optimal values characterize key quantities of the approximation error. We have provided a systematic way to encode the maximal gain of a given \gls{MPC} law through binary and continuous variables subject to \gls{MI} constraints. This optimization-based result is compatible with existing results from the machine learning literature on computing the Lipschitz constant of a trained \gls{ReLU} network.  Taken together they provide sufficient conditions to assess the reliability, in terms of stability of the closed-loop system, of a given \gls{ReLU}-based approximation of an \gls{MPC} scheme.

We believe our work can be extended in several ways. An interesting direction to explore is the connection between our results and those established in \cite{karg2020efficient}, which could allow one to apply firm limits to the complexity of the underlying \gls{ReLU}-based approximation with provable stability guarantees. Since the analysis carried out in \S \ref{sec:stability} provides only sufficient conditions for good controller performance, it would be interesting to investigate whether there exist less conservative conditions involving different, but still computable, properties of the approximation error. Finally, since our results involve the optimal values of two \glspl{MILP}, it may be of interest to explore probabilistic counterparts of our deterministic statements, e.g., via PAC learning and randomized approaches.


\appendix

\subsection{Proof of \S \ref{sec:stability}}\label{app:proofs}

\textit{Proof of Lemma~\ref{lemma:NN_stability}}: Define the process $w(x) \eqdef B e(x)$ in \eqref{eq:perturbed_dyn}, an additive disturbance taking values in $\mc{W} \eqdef \{w \in \R^n \mid \|w\|_2 \leq \bar{w}\}$ for all $k \in \N$, with $\bar{w} \eqdef s \bar{e}_{\alpha}$, where $s > 0$ is a scaling factor accounting for $\|B\|_2$ and the choice of $\alpha \in \N_\infty$.


In view of Standing Assumption~\ref{standing:mp-QP}.(i), the nominal closed-loop system $x^+ = Ax + Bu_{\textrm{MPC}}(x)$ converges exponentially to the origin with region of attraction $\mc{X}$ \cite[\S 2.5.3.1]{rawlings2017model}.  In this case, the function $V_T(\cdot)$ serves as a Lyapunov function for which there exist constants $c_2 > c_1 > 0$ satisfying
\begin{equation}\label{eq:ISS_Lyap_1}
	\begin{aligned}
		&c_1 \|x\|_2^2 \leq V_T(x) \leq  c_2 \|x\|_2^2,\\
		&V_T(Ax + Bu_{\textrm{MPC}}(x)) \leq  V_T(x) - c_1 \|x\|_2^2,
	\end{aligned}
\end{equation}
for all $x \in \mc{X}$ \cite{mayne2000constrained}. There also exists a contraction factor $\gamma \eqdef (1 - c_1/c_2) \in (0,1)$ such that $V_T(Ax + Bu_{\textrm{MPC}}(x)) \leq  \gamma V_T(x)$ for all $x \in \mc{X}$.
Now, let $\Omega_c$ denote the largest sublevel set of $V_T(\cdot)$ contained in $\mc{X}$, with $c \eqdef \textrm{max} \, \{a \geq 0 \mid \Omega_a \subseteq \mc{X} \}$.
It follows from \cite[Prop.~7.13]{rawlings2017model} that $V_T(\cdot)$ is Lipschitz continuous in $\Omega_c$ with constant $\mc{L}_2(V_T, \Omega_c)$.
Due to the presence of the disturbance, however, the value function is not guaranteed to decrease along the trajectories of \eqref{eq:perturbed_dyn}, since we have
\ifTwoColumn
\else
\begin{equation}\label{eq:ISS_Lyap_2}
		V_T(Ax + Bu_{\textrm{MPC}}(x) + w) - V_T(x) \leq V_T(Ax + Bu_{\textrm{MPC}}(x)) - V_T(x) + \mc{L}_2(V_T, \Omega_c) \|w\|_2,
\end{equation}
\fi
for all $w \in \mc{W}$, and hence $V_T(Ax + Bu_{\textrm{MPC}}(x) + w) \leq  \gamma V_T(x) + \mc{L}_2(V_T, \Omega_c) \|w\|_2$. From \eqref{eq:ISS_Lyap_1} and \eqref{eq:ISS_Lyap_2}, note that $V_T(\cdot)$ is a candidate \gls{ISS}-Lyapunov function \cite[Def.~3.2]{jiang2001input} in $\Omega_c$ for the perturbed system in \eqref{eq:perturbed_dyn}.
To prove that the system is \gls{ISS} in $\Omega_c$ \cite[Def.~3.1]{jiang2001input}, it remains only to show that $\Omega_c$ is robust positively invariant \cite[Def.~4.3]{blanchini2015set}.
Since $\mc{W}$ is bounded, we can first focus on a lower sublevel set of $\Omega_c$, say $\Omega_b$ with $b < c$, and show that it is robust positively invariant for the dynamics in \eqref{eq:perturbed_dyn}. To prove this, let $b$ be chosen such that for a given $x \in \Omega_b$, $V_T(Ax + Bu_{\textrm{MPC}}(x) + w) < V_T(x) \leq b$, for all $w \in \mc{W}$. We then have
$
V_T(Ax + Bu_{\textrm{MPC}}(x) + w) \leq  \gamma V_T(x) + \mc{L}_2(V_T, \Omega_c) \|w\|_2 \leq \gamma b + \bar{w} \mc{L}_2(V_T, \Omega_c),
$
which is strictly smaller than $b$ if $b > \bar{w} \mc{L}_2(V_T, \Omega_c)/(1-\gamma)$. Thus, $(Ax + Bu_{\textrm{MPC}}(x) + w) \in \Omega_b$, for all $x \in \Omega_b$ and $w \in \mc{W}$.


We now suppose that the disturbance is bounded by $\bar{w} \leq (\rho - \gamma) b / \mc{L}_2(V_T, \Omega_c)$ for some $\rho > \gamma$. Such a restriction amounts to a condition that $\Omega_c$ is a contractive set, and hence robust positively invariant, w.r.t. the dynamics \eqref{eq:perturbed_dyn}. In fact, it first yields $b \geq \bar{w} \mc{L}_2(V_T, \Omega_c) / (\rho - \gamma)$. To satisfy the chain of inequalities $b \geq \bar{w} \mc{L}_2(V_T, \Omega_c) / (\rho - \gamma) > \bar{w} \mc{L}_2(V_T, \Omega_c)/(1-\gamma)$, thus guaranteeing the robust invariance of $\Omega_b$, we therefore require $\rho \in (\gamma, 1)$. Moreover, since $\bar{w} \leq (\rho - \gamma) b / \mc{L}_2(V_T, \Omega_c)$, by noting that, for all $x \in \Omega_c \setminus \Omega_b$, $V_T(x) \geq b$, we have
$
V_T(Ax + Bu_{\textrm{MPC}}(x) + w) \leq  \gamma V_T(x) + \mc{L}_2(V_T, \Omega_c) \|w\|_2 \leq \gamma V_T(x) + \bar{w} \mc{L}_2(V_T, \Omega_c)
\leq \gamma V_T(x) + (\rho - \gamma) b
\leq \gamma V_T(x) + (\rho - \gamma) V_T(x)
\leq \rho V_T(x),
$
showing that $\Omega_c$ is contractive for \eqref{eq:perturbed_dyn}. Therefore, for any $x(0) \in \Omega_c$ the perturbed system enters, in finite time, the robust invariant set $\Omega_b$.
However, the magnitude of the disturbance $\bar{w}$ can not be arbitrary since we  assumed $b < c$, i.e., $\bar{w} < (\rho - \gamma) c/\mc{L}_2(V_T, \Omega_c)$.  This in turn implies
\begin{equation}\label{eq:error_bound}
	\bar{e}_{\alpha} < \frac{(\rho - \gamma) c}{s \mc{L}_2(V_T, \Omega_c)} \reqdef \zeta.
\end{equation}
The proof is completed from \cite[Lemma~B.38]{rawlings2017model} by noting that $\Omega_c$ is robust invariant for \eqref{eq:perturbed_dyn} with \gls{ISS}-Lyapunov function $V_T(\cdot)$, thus ensuring that the perturbed dynamics \eqref{eq:perturbed_dyn} is \gls{ISS} in $\Omega_c$.
\hfill\qedsymbol

\smallskip

\textit{Proof of Lemma~\ref{lemma:NN_stability_local}}: For any $x \in \mc{X}_\infty$, we know that $u_\textrm{MPC}(x) = \bar{K}_\textrm{MPC} x$, and in view of Standing Assumption~\ref{standing:mp-QP}.(i), the closed-loop matrix $\bar{A} \eqdef A + B \bar{K}_\textrm{MPC}$ is Schur stable. Thus, there exist constants $\varsigma > 0$ and $\lambda \in (0,1)$ such that $\|\bar{A}^k\|_2 \leq \varsigma \lambda^k$, $k \in \N_0$. With this consideration, for all $k \in \N$ and $x(0) \in \mc{X}_\infty$, the evolution of the perturbed system \eqref{eq:perturbed_dyn} satisfies the following relations:
$$
\begin{aligned}
	\|x(k)\|_2 &= \| \bar{A}^k x(0) + \textstyle\sum_{j  = 0}^{k-1} \bar{A}^{k - j - 1} B e(j) \|_2\\
	&\leq \varsigma \lambda^{k} \|x(0)\|_2 + \textstyle\sum_{j  = 0}^{k-1} \|\bar{A}^{k - j - 1}\|_2 \|B e(j)\|_2\\
	&\leq \varsigma \lambda^{k} \|x(0)\|_2 + s' \varsigma \mc{L}_\alpha(e, \mc{X}_\infty)  \textstyle\sum_{j  = 0}^{k-1} \lambda^{k - j - 1} \|x(j)\|_2,
\end{aligned}
$$
where we have exploited the Lipschitz continuity of $e(\cdot)$ on $\mc{X}_\infty$, as well as the fact that $e(0) = 0$ (otherwise the origin is not an equilibrium for \eqref{eq:LTI_sys}). Here, $s' > 0$ is a scaling factor that accounts for $\|B\|_2$ and the choice of $\alpha \in \N_\infty$.
Then, by introducing $\upsilon(k) \eqdef \lambda^{-k} x(k)$, the inequality above becomes
$
\|\upsilon(k)\|_2 \leq \varsigma \|\upsilon(0)\|_2 + s' \varsigma \mc{L}_\alpha(e, \mc{X}_\infty)  \textstyle\sum_{j  = 0}^{k-1} \lambda^{-1} \|\upsilon(j)\|_2
$,
and by leveraging the Gr\"onwall inequality \cite[Cor.~4.1.2]{agarwal2000difference}, we obtain
$
\|\upsilon(k)\|_2 \leq \varsigma \|\upsilon(0)\|_2 \textstyle\prod_{j  = 0}^{k-1} (1 + s' \varsigma  \mc{L}_\alpha(e, \mc{X}_\infty)  \textstyle\sum_{j  = 0}^{k-1} \lambda^{-1}) \leq \varsigma \|\upsilon(0)\|_2 \, \mathrm{exp}(k s' \varsigma \mc{L}_\alpha(e, \mc{X}_\infty) \lambda^{-1})
$,
or, equivalently,
$$
	\begin{aligned}
		\|x(k)\|_2 &\leq \varsigma \lambda^k \|x(0)\|_2 \, \mathrm{exp}(k s' \varsigma \mc{L}_\alpha(e, \mc{X}_\infty) \lambda^{-1})\\
		&= \varsigma \|x(0)\|_2 \, \mathrm{exp}(k s' \varsigma \mc{L}_\alpha(e, \mc{X}_\infty) \lambda^{-1} + k \ln \lambda)\\
	\end{aligned}
$$
Then, if $s \varsigma  \mc{L}_\alpha(e, \mc{X}_\infty) \lambda^{-1} + \ln \lambda < 0$, which leads to
\begin{equation}\label{eq:upper_bound_lipschitz}
	\mc{L}_\alpha(e, \mc{X}_\infty) < -\frac{\lambda \ln \lambda}{s' \varsigma} \eqqcolon \vartheta,
\end{equation}
the system \eqref{eq:perturbed_dyn} is exponentially stable in $\mc{X}_\infty$, and so is the \gls{LTI} system in \eqref{eq:LTI_sys} with \gls{PWA-NN} controller $u = u_{\normaltext{\textrm{NN}}}(x)$.
\hfill\qedsymbol

\smallskip

\textit{Proof of Theorem~\ref{th:exp_conv}}: 	We consider the case in which $x \in \Omega_c  \setminus \Omega_b$, since if $x \in \Omega_b$ and the inclusion $\Omega_b \subseteq \mc{X}_\infty$ holds, then the conclusion follows immediately from Lemma~\ref{lemma:NN_stability_local}.
Thus, by focusing on the perturbed dynamics in \eqref{eq:perturbed_dyn}, if $\bar{e}_\alpha$ satisfies \eqref{eq:error_bound}, for any $x \in \Omega_c  \setminus \Omega_b$, $V_T(Ax + Bu_{\textrm{MPC}}(x) + w) \leq \rho V_T(x)$ for all $w \in \mc{W}$ and $\rho \in (\gamma,1)$.
Thus, with $x(k+1) = Ax(k) + Bu_{\textrm{MPC}}(x(k)) + w(k)$ leveraging the relations in \eqref{eq:ISS_Lyap_1} yields:
	$
		c_1 \|x(k+1)\|_2^2 \le V_T(Ax(k)+B u_{\text{MPC}}(x(k))+w) \le \rho V_T(x(k)) \le \ldots \le \rho^{k+1} V_T(x(0)) \le c_2 \rho^{k+1} \|x(0)\|^2_2,
	$
for any possible realization of the sequence of $\{w(0), w(1), \ldots, w(k)\}$ leading to $x(k+1)$ by starting from some $x(0) \in \Omega_c \setminus \Omega_b$. This therefore implies that $\|x(k)\|_2 \le \sqrt{c_2/c_1} (\sqrt{\rho})^k \|x(0)\|_2 = \sqrt{c_2/c_1} \|x(0)\|_2 \ \mathrm{exp}(k \ln \sqrt{\rho})$
%
for all $k \in \N$ such that $k \leq k_1$, where $k_1 \in \N$ denotes the time instant in which \eqref{eq:perturbed_dyn} enters $\Omega_b$, which is guaranteed to exist finite in view of Lemma~\ref{lemma:NN_stability}. Then, if we can choose $b \geq 0$ such that $\Omega_b \subseteq \mc{X}_\infty$,
there also exists some $k_2 \leq k_1$, $k_2 \in \N$, in which the perturbed dynamics in \eqref{eq:perturbed_dyn} enters $\mc{X}_\infty$, and hence exponentially converges to the origin.
We thus conclude that the origin is an exponentially stable equilibrium for \eqref{eq:LTI_sys} with \gls{PWA-NN} controller $u = u_{\normaltext{\textrm{NN}}}(x)$.
\hfill\qedsymbol

\subsection{Further numerical performance results}\label{app:maximal_gain_MPC}
We compare the computation of $\mc{L}_{\alpha}(u_\normaltext{\textrm{MPC}},\mc{X})$, $\alpha \in \{1, \infty\}$ using the proposed optimization-based approach described in Theorem~\ref{th:norm_comp}, relative to one based on direct computation via the \gls{MPC} Toolbox \cite{bemporad2021model}.  We test against various numerical examples available in the literature, with all models summarized in Table~\ref{tab:examples}.  Numerical results are shown in Table~\ref{tab:numerical_values}.

In all examples with the exception of Example~3, the computational time required to solve the \normaltext{\gls{MILP}} in Theorem~\ref{th:norm_comp} is lower than that required by the \normaltext{\gls{MPC}} Toolbox to generate a solution. When computing the maximal gain, which requires comparison of linear gains across all regions of the partition of $u_\normaltext{\textrm{MPC}}(\cdot)$, \normaltext{\gls{eMPC}} requires significant memory to store all the involved matrices/vectors for every region (up to 729 in Ex.~6).
%
%
%
%
Finally, we note that almost all the numerical values reported in the columns $\mc{L}_{\alpha}(u_\normaltext{\textrm{MPC}},\mc{X})$, $\alpha \in \{1, \infty\}$ agree  between the two methods. In those cases showing a discrepancy, e.g., Ex.~1, 2, and 6, the ``\normaltext{MIPGap}'' columns suggest that, in computing the critical regions and associated  controllers, the \normaltext{\gls{MPC}} Toolbox makes some internal approximation, as the primal-dual gap of the solutions to the \normaltext{\gls{MILP}} is exactly zero, and hence those solutions are necessarily optimal.

\begin{table*}[!t]
	\caption{Examples from the literature}\label{tab:examples}
	\centering
		\begin{tabular}{ccccccccc}
			\toprule
			Ex. & Reference & $A$ & $B$ & $\mc{X}$ & $\mc{U}$ & $Q$ & $R$ & $T$ \\
			\midrule
			1& \footnotesize{\cite[Ex.~2.26]{schulze2014numerical}} & $\begin{bmatrix}
				\phantom{-}1.1 & 0.2 \\
				-0.2 & 1.1
			\end{bmatrix}$ & $\begin{bmatrix}
				0.5 & 0\\
				0 & 0.4
			\end{bmatrix}$ & $\begin{aligned}
				&|x_1| \leq 5\\
				&|x_2| \leq 5
			\end{aligned}$ & $\begin{aligned}
				&|u_1| \leq 1\\
				&|u_2| \leq 1
			\end{aligned}$ & $I$ & $0.1  I$ & $3$ \\
			\midrule
			2& \footnotesize{\cite[Rem.~4.8]{gutman1987algorithm}}  & $\begin{bmatrix}
				1 & 0.5 & 0.125\\
				0 & 1 & 0.5\\
				0 & 0 & 1
			\end{bmatrix}$ & $\begin{bmatrix}
				0.02\\
				0.125\\
				0.5
			\end{bmatrix}$ & $\begin{aligned}
				&|x_1| \leq 20\\
				&|x_2| \leq 3\\
				&|x_3| \leq 1
			\end{aligned}$ & $|u| \leq 0.5$ & $I$ & $1$ & $3$ \\
			\midrule
			3& \footnotesize{\cite[Ex.~3]{darup2016some}} & $\begin{bmatrix}
				0 & 1 \\
				1 & 0
			\end{bmatrix}$ & $\begin{bmatrix}
				2\\
				4
			\end{bmatrix}$ & $\begin{aligned}
				&|x_1| \leq 5\\
				&|x_2| \leq 5
			\end{aligned}$ & $|u| \leq 1$ & $I$ & $4.5$ & $8$ \\
			\midrule
			4& \footnotesize{\cite[Eqs.~(2.8)--(2.9)]{gutman1987algorithm}} & $\begin{bmatrix}
				1 & 1\\
				0 & 1
			\end{bmatrix}$ & $\begin{bmatrix}
				0.5\\
				1
			\end{bmatrix}$ & $\begin{aligned}
				&|x_1| \leq 25\\
				&|x_2| \leq 5
			\end{aligned}$ & $|u| \leq 1$ & $I$ & $0.1$ & $10$\\
			\midrule
			5& \footnotesize{\cite[Ex.~6.1]{bemporad2003suboptimal}} & $\begin{bmatrix}
				0.7969 & -0.2247 \\
				0.1798 & \phantom{-}0.9767
			\end{bmatrix}$ & $\begin{bmatrix}
				0.1271\\
				0.0132
			\end{bmatrix}$ & $\begin{aligned}
				&|x_1| \leq 4\\
				&|x_2| \leq 4
			\end{aligned}$ & $|u| \leq 1$ & $I$ & $0.1$ & $8$\\
			\midrule
			6 & \footnotesize{\cite[\S VI]{jones2009approximate}} & $\begin{bmatrix}
				1 & 1 \\
				0 & 1
			\end{bmatrix}$ & $\begin{bmatrix}
				0.42 & 0.9\\
				0.38 & 0.67
			\end{bmatrix}$ & $\begin{aligned}
				&|x_1| \leq 40\\
				&|x_2| \leq 10
			\end{aligned}$ & $\begin{aligned}
				&|u_1| \leq 0.1\\
				&|u_2| \leq 0.1
			\end{aligned}$ & $I$ & $30 I$ & $10$\\
			\midrule
			7 & \footnotesize{\cite[\S IV]{markolf2021polytopic}} & $\begin{bmatrix}
				1.5 & \phantom{-}0 \\
				1 & -1.5
			\end{bmatrix}$ & $\begin{bmatrix}
				1 & 0\\
				0 & 1
			\end{bmatrix}$ & $\begin{aligned}
				&|x_1| \leq 6\\
				&|x_2| \leq 6
			\end{aligned}$ & $\begin{aligned}
				&|u_1| \leq 5\\
				&|u_2| \leq 5
			\end{aligned}$ & $I$ & $I$ & $12$\\
			\bottomrule
		\end{tabular}
	\end{table*}

	\begin{table*}[!t]
		\caption{Numerical results for the computation of $\mc{L}_{\alpha}(u_\textrm{MPC},\mc{X})$, $\alpha \in \{1, \infty\}$ -- Examples in Table~\ref{tab:examples}}\label{tab:numerical_values}
		\centering
			\begin{tabular}{cccccccccccc}
				\toprule
				\multirow{2}[2]{*}{Ex.} & \multicolumn{4}{c}{\gls{MPC} Toolbox \cite{bemporad2021model}} & & \multicolumn{6}{c}{\gls{MILP} Th.~\ref{th:norm_comp} + Prop.~\ref{prop:MILP_matrix_norm}} \\
				\cmidrule{2-5} \cmidrule{7-12} & \multicolumn{1}{c}{\footnotesize{\# of $\mc{R}_{\mc{A}}$}} & \multicolumn{1}{c}{\footnotesize{CPU time [s]}} & \multicolumn{1}{c}{\footnotesize{$\mc{L}_{1}$}} & \multicolumn{1}{c}{\footnotesize{$\mc{L}_{\infty}$}} &  & \multicolumn{1}{c}{\footnotesize{$\mc{L}_{1}$}} & \multicolumn{1}{c}{\footnotesize{CPU time [s]}} & \multicolumn{1}{c}{\footnotesize{MIPGap}} & \multicolumn{1}{c}{\footnotesize{$\mc{L}_{\infty}$}} & \multicolumn{1}{c}{\footnotesize{CPU time [s]}} & \multicolumn{1}{c}{\footnotesize{MIPGap}}\\
				\midrule
				1 & 99 & 1.24 & {16.6} & {11.67} & & {16.1} & 0.51& 0& {11.7} & 0.45& 0\\
				\midrule
				2 & 111 & 1.56 & {11.98} &  {7.99} && {12.00} & 0.56 & 0 & {8.00} & 0.69 & 0\\
				\midrule
				3 & 27 & 1.18 & {0.5} & {0.49} & &  {0.5} & 2.48 & 0 & {0.5} & 3.89 & 0\\
				\midrule
				4 & 317 & 23 & {1.88} & {1.27} & & {1.88} & 10.8 & 0 & {1.27} & 16.1 & 0\\
				\midrule
				5 & 105 & 1.34 & {3.1} & {2.39} & & {3.1} & 0.99 & 0 & {2.39} & 0.98 & 0\\
				\midrule
				6 & 729 & 117.8 & {1.76} & {1.52} & & {1.77} & 8.4 & 0 & {1.53} & 7.9 & 0\\
				\midrule
				7 & 15 & 5.46 & {1.66} & {1.66} & & {1.66} & 0.14 & 0 & {1.66} & 0.14 & 0\\
				\bottomrule
			\end{tabular}
		\end{table*}

\bibliographystyle{IEEEtran}
\bibliography{PWA_NN_controllers.bib}



\end{document}